\documentclass[journal]{IEEEtran}
\usepackage{amsfonts}
\usepackage{amsfonts}
\usepackage{amsfonts}
\usepackage{mathbbold}
\usepackage{amsfonts}
\usepackage{amssymb}
\usepackage{stfloats}
\usepackage{cite}
\usepackage{graphicx}
\usepackage{psfrag}
\usepackage{subfigure}
\usepackage{amsmath}
\usepackage{array}
\begin{document}

\title{Delay-Sensitive Distributed Power and Transmission Threshold Control for
S-ALOHA Network with Finite State Markov Fading Channels}

\newtheorem{Thm}{Theorem}
\newtheorem{Lem}{Lemma}
\newtheorem{Cor}{Corollary}
\newtheorem{Def}{Definition}
\newtheorem{Exam}{Example}
\newtheorem{Alg}{Algorithm}
\newtheorem{Prob}{Problem}
\newtheorem{Rem}{Remark}
\newtheorem{Proof}{Proof}
\newtheorem{Subproblem}{Subproblem}

\author{\authorblockN{Huang Huang {\em Student Member, IEEE} and Vincent K. N. Lau, {\em Senior Member, IEEE}}
\thanks{Manuscript received Dec 19, 2008; revised Jun 12, 2009 and Aug 15,
2009; accepted Aug 15, 2009. The editor coordinating the review of
this paper and approving it for publication was Sonia Aissa. This
work is supported by RGC 615407. The material in this paper was
presented in part at the IEEE International Symposium on Information
Theory, Seoul, Korea, June/July 2009.}
\thanks{ The authors are with
the Department of Electronic and Computer Engineering (ECE), The
Hong Kong University of Science and Technology (HKUST), Hong Kong.
(email: \{huang, eeknlau\}@ust.hk).}}

\markboth{To be appeared in IEEE Trans. Wireless Commun.}%
{Shell \MakeLowercase{\textit{et al.}}: Bare Demo of IEEEtran.cls
for Journals}
%

\maketitle

\begin{abstract}
In this paper, we consider the delay-sensitive power and
transmission threshold control design in S-ALOHA network with FSMC
fading channels. The random access system consists of an access
point with $K$ competing users, each has access to the local channel
state information (CSI) and queue state information (QSI) as well as
the common feedback (ACK/NAK/Collision) from the access point. We
seek to derive the delay-optimal control policy (composed of
threshold and power control). The optimization problem belongs to
the memoryless policy $K$-agent infinite horizon decentralized
Markov decision process (DEC-MDP), and finding the optimal policy is
shown to be computationally intractable. To obtain a feasible and
low complexity solution, we recast the optimization problem into two
subproblems, namely the {\em power control} and the {\em threshold
control} problem. For a given threshold control policy, the power
control problem is decomposed into a {\em reduced state MDP} for
single user so that the overall complexity is $\mathcal{O}(NJ)$,
where $N$ and $J$ are the buffer size and the cardinality of the CSI
states. For the threshold control problem, we exploit some special
structure of the collision channel and common feedback information
to derive a low complexity solution. The delay performance of the
proposed design is shown to have substantial gain relative to
conventional throughput optimal approaches for S-ALOHA.
\end{abstract}


\section{Introduction}\label{sec:intro}
Random access network is a hot research topic due to its robustness
in system performance. In particular, ALOHA is a popular example of
random access protocol which has attracted a lot of research
attention over the past two decades. One important application is
the access network (such as the infrastructure mode in WiFi) where
multiple nodes compete for transmission opportunity to transmit data
to an access point (AP). In \cite{Tsybakov:1979}, the authors
considered the design and analysis of the traditional buffered
slotted ALOHA (S-ALOHA) in which finite users with infinite buffer
attempt to transmit a backlogged packet according to a {\em
transmission probability} in one slot, and the packet is
successfully received if and only if exact one packet is
transmitted. In asymmetric network (heterogenous users), the
stability region has only been obtained in two and three user
cases\cite{Stability:1999}. The study of the stability region for
general number of users is difficult because the transition
probability of the state space of the interacting queues alters from
the non-empty to empty buffer case. In \cite{Dominant:System}, the
authors proposed a {\em dominant system} technique to obtain a lower
bound for the stability region for the general case. In symmetric
ALOHA network (homogeneous users), all users are statistically
identical and hence, the stability region is degenerated to one
dimension. It is shown in \cite{Tsybakov:1979,CSI_RA:2005} that the
system is stable as long as the arrival rate is less than the
average throughput. As a result, stability analysis is equivalent to
the throughput analysis. The authors in \cite{CSI_RA:2005} extended
the protocol to an adaptive ALOHA over the multi-packet reception
(MPR) channel to maximize the system throughput. For instance, the
transmission probability is a function of the local channel state
information (CSI). In \cite{Thrput:pwr}, the authors extended to the
adaptive transmission rate and power control w.r.t to CSI to
maximize the throughput. In \cite{Yu:2006}, it is shown that a
simple adaptive permission probability scheme, namely binary
scheduling, is throughput optimal for homogeneous users with
adaptive transmission rate in collision channel. In the binary
scheduling scheme, there is a transmission threshold in which user
could attempt to transmit its backlogged packet only when its local
CSI exceeds the threshold.

In all the above works on stability and throughput analysis and
optimization, the delay performance has been ignored completely. In
practice, applications are delay-sensitive and it is critical to
optimize the delay performance in S-ALOHA network to support
realtime applications. In \cite{Delay:2005}, the authors surveyed
the recent works on delay analysis of traditional S-ALOHA network in
which exact delay can be obtained only in two user case. In
\cite{TUA:2007}, the delay performance for finite user finite buffer
is analyzed using the tagged user analysis (TUA) method. Although
the channel fading is considered, adaptive transmission probability
and rate with power control is not allowed. In
\cite{Delay:trade-off}, the trade-off between delay and energy in
additive write Gaussian noise (AWGN) channel with no queue state
information (QSI) is investigated. However, they assumed
multi-access coding to ensure successful reception for each user
even if all competing users transmit simultaneously. In
\cite{Yeh:01PhD}, the authors proved that the longest queue highest
possible rate (LQHPR) policy, which is a centralized control policy
requiring perfect knowledge of global QSI and global CSI, is
delay-optimal in symmetric network. While the above works deal with
the delay performance of S-ALOHA network, there are still a lot of
technical challenges to be solved. They are listed below.

\begin{itemize}
\item{\bf Queue-aware power and threshold control for S-ALOHA:}
Previous literature focused either on the power control (under a
fixed and common threshold for all users) for throughput
optimization, or on the delay analysis of uncontrolled S-ALOHA
network. Both the transmission threshold control and power control
policies are important means to optimize the delay performance of
S-ALOHA. However, due to the lack of global knowledge on CSI and
QSI, it is quite challenging to design delay-sensitive control
schemes for S-ALOHA networks.
\item{\bf Exploiting memory in the fading channels:}
Existing works have assumed memoryless adaptation in which the
control actions are done independently slot by slot (assuming fading
is i.i.d). While i.i.d fading could lead to simple solution, it
fails to exploit the memory of the time varying fading channels,
which is critical to boost the delay performance of S-ALOHA network.

\item{\bf Utilization of local QSI and common feedback information from the AP:}
Existing control policy on throughput optimization only adapts to
the local CSI and did not exploit the local QSI as well as common
feedback information from the AP. These side information are also
critical to improve the delay performance of the S-ALOHA network.
\end{itemize}

In this paper, we shall propose a delay-sensitive power and
transmission threshold control algorithm for S-ALOHA network which
addresses the above three important issues. We consider a S-ALOHA
network with $K$ users. The transmit power and threshold control
policies adapt to the local CSI, local QSI as well as common
feedback information (ACK/NAK/Collision) from the AP. The
delay-optimization problem belongs to the memoryless policy
$K$-agent infinite horizon decentralized Markov decision process
(DEC-MDP)\cite{DEC-MDP:2002}. The problem of finding the optimal
policy is proved to be {\em
NP-hard}\cite{Littman:1994,Meuleau:1999}, which means that the
optimal solution is computationally intractable. To obtain a
feasible and low complexity solution, we recast the optimization
problem into two subproblems, namely the {\em power control} and the
{\em threshold control} problem. For a given threshold control
policy, the power control problem is decomposed into a {\em reduced
state MDP} for single user so that the overall complexity is
$\mathcal{O}(NJ^2)$, where $N$ and $J$ are the buffer size and the
cardinality of the CSI states. On the other hand, we solve the
threshold control problem by exploiting the special structure of the
S-ALOHA network and common feedback information to derive a low
complexity solution. The delay performance of the proposed design is
shown to have substantial gain relative to conventional solutions.

This paper is organized as follows. In section \ref{sec:model}, we
outline the system model of S-ALOHA network and define the
delay-optimal control policy. In section \ref{sec:problem}, we shall
formulate the delay-optimal problem and introduce the DEC-MDP model.
In section \ref{sec:symm}, we exploit the special structure in
symmetric network. We also extend to asymmetric case in section
\ref{sec:asymmetric} and illustrate the performance via simulations
in section \ref{sec:simulation}. A brief summary is given in section
\ref{sec:summary} finally.

\section{System Model}\label{sec:model}
In this section, we shall elaborate the system model, including
source and physical layer model, as well as the control policy in
symmetric network, and extend to the asymmetric case in section
\ref{sec:asymmetric}. We consider a $K$ users S-ALOHA network in
this paper. The time dimension is partitioned into {\em slots} (each
slot lasts $\tau$ seconds). The $m$-th slot means the time interval
$(m\tau,(m+1)\tau)$, $m=0,1,2\cdots$. Fig. \ref{fig:system}
illustrates the top level system model in symmetric network. The $K$
competing users are coupled together via the transmission threshold
and power control policy.
\begin{figure}
 \begin{center}
  \resizebox{8cm}{!}{\includegraphics{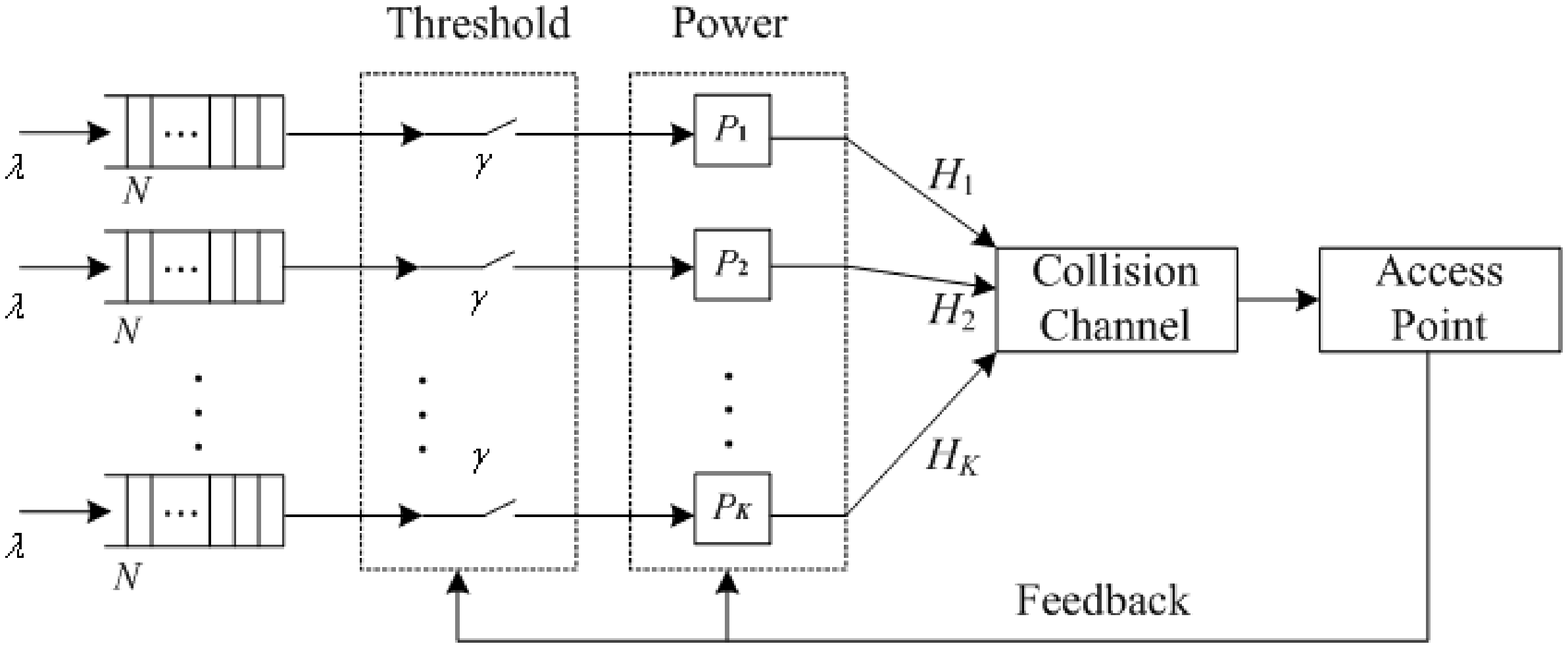}}
 \end{center}
    \caption{The system model in symmetric S-ALOHA network.}
    \label{fig:system}
\end{figure}

\subsection{Source Model}
\label{sec:source_model} For simplicity, the arrival packet rate of
all the users is assumed to follow independent Poisson distribution
with arrival rates $\lambda$ (number of packets per second). The
packet length of the data source $N_b$, follows exponential
distribution with mean packet size $\overline{N_b}$ (bits per
packet), and the buffer size is $N$ (packets). The QSI of the whole
system at the $m$-th slot is denoted by
$\mathbf{Q}_m=\{Q_{k,m}\}_{k=1}^K\in\mathcal{N}^K$, where $Q_{k,m}$
is the number of packets in the $k$-th user's buffer, and
$\mathcal{N}=\{0,1,2,...,N\}$ denotes a finite state space of local
QSI for single user. When the buffer is full, i.e, $Q_{k,m}=N$, it
will not accept any potential new packets.

\subsection{Physical Layer Model and Feedback Mechanism}
We consider a block fading channel between each user and the AP. The
CSI at $m$-th slot is denoted by
$\mathbf{H}_m=\{H_{k,m}\}_{k=1}^K\in\mathcal{S}^K$, where $H_{k,m}$
is the channel gain for user $k$, and
$\mathcal{S}=\{S_{i}\}_{i=1}^{J}$ denote a set of $J$ CSI states for
single user. $\{H_{k,m}\}_{m=1}^{\infty}$ is modeled as a stationary
ergodic process\cite{FSMC:Channel}, which is independent among
users. Specifically, let $p_{i,j}=\Pr\{H_{k,m}=S_j|H_{k,m-1}=S_i\}$
be the state transition probability and
$\pi_j=\Pr\{H_{k,\infty}=S_j\}$ be the stationary probability. All
the users share a common spectrum with a bandwidth of $W$Hz using
S-ALOHA protocol. The signal received by the AP at $m$-th slot is
given by:
\begin{equation}
\label{eq:sys_model}
y[m]=\sum\nolimits_{k=1}^{K}\sqrt{H_{k,m}}x_{k}[m]+z[m]
\end{equation}
where $x_{k}[m]$ is the transmit signal for the $k$-th user at
$m$-th slot, and $\{z[m]\}_{m=1}^{\infty}$ is the i.i.d
$\mathcal{N}(0,N_0)$ noise. Suppose that only the $k$-th user
attempts to transmit its packet to the AP at the $m$-th  slot. The
maximum achievable data rate (b/s) of the $k$-th user is given by:
\begin{equation}
\label{eq:rate}
R(P_{k,m},H_{k,m})=W\log_2\left(1+\frac{P_{k,m}H_{k,m}}{N_0W}\right)
\end{equation}
where $P_{k,m}$ and $H_{k,m}$ is the power and channel gain of
$k$-th user at $m$-th slot.

To decouple the delay-optimal design from the detailed
implementation of the modulation and coding in the physical layer,
we assumed that the data rate (\ref{eq:rate}) is achievable. In
fact, it has been shown \cite{LDPC:2001} that the Shannon's limit in
(\ref{eq:rate}) can be achieved to within $0.05$dB SNR using LDPC
with $2$K byte block size at $1\%$ PER. We consider a collision
channel for the S-ALOHA random access and hence, the AP could only
decode the data successfully when there is only one user
transmitting in any time slot. At the end of each slot, the AP
broadcasts the ACK/NAK/Collision feedback, denoted as
$\mathcal{Z}=(1,0,e)$ \cite{Bertsekas:1992}, to all the $K$ users in
the network. For instance, ACK $(Z=1)$ means that exactly one user
has transmitted the packet, and data was successfully decoded; NAK
$(Z=0)$ means that none of users has transmitted and hence, no data
was received; Collision ($Z=e$) means that at least two users have
transmitted, and the data was corrupt\footnote{Since
we assume strong coding is used by each user, we ignore the case with transmission
error.}.

\subsection{Control Policy}\label{sec:con_pol}
Each user decides whether to transmit a packet at the beginning of a
slot using a {\em threshold mechanism}. Due to symmetry, a user will
transmit if the buffer is not empty and its local CSI exceeds a
common system threshold $\gamma_{m}$\footnote{In symmetric network, users are
statistically identical (e.g. same fading channel, same arrival packet rate and same average power constraint)
and a common threshold is reasonable for fairness consideration (achieving the same average delay performance).
On the other hand, for the asymmetric network, we have considered the flexibility of different thresholds
for different users (because the users are not statistically identical anymore).}. If there are more than one
backlogged users' local CSI exceeding the threshold, then collision
will occur and none of the packets could get through. As a result,
$\gamma_{m}$ determines the priority on the access opportunity of
each user. In this paper, we shall consider an adaptive threshold
control to exploit the fading memory to minimize the system delay. A
{\em stationary threshold control policy} $\pi_\gamma$ is defined
below:
\begin{Def}[Stationary Threshold Control Policy]\footnote{We have assumed the deterministic threshold control
policy here. In fact, the same formulation and approach can be used
to deal with a {\em transmission probability} approach rather than
threshold approach. The users will transmit in a probability at different CSI state
according to a probability function $\varphi(H)\in [0,1]$. The transmission
control policy is defined as
$\pi_{\varphi}(\varphi_{m-1},Z_{m-1})=\varphi_{m}$, i.e, mapping
from the common information to current slot's transmission
probability function.}
A stationary threshold control policy $\pi_{\gamma}:
\mathcal{S}\times\mathcal{Z}\rightarrow \mathcal{S}$ is defined as
the mapping from the previous slot's system threshold $\gamma_{m-1}$
and common feedback $Z_{m-1}$ from the AP to the system threshold
$\pi_{\gamma}(\gamma_{m-1},Z_{m-1})=\gamma_{m}$ in current slot. The
set of all feasible stationary policies $\pi_{\gamma}$ is denoted as
$\mathcal {P}_{\gamma}=\left\{\pi_{\gamma}:
\pi_{\gamma}(\gamma_{m-1},Z_{m-1})\in\mathcal{S}\right\}$.
\end{Def}

The threshold control is adaptive to the common information for all
the $K$ users and hence, each user could determine the system
threshold just from the feedback from the AP.

Denote $\boldsymbol{\chi}_m=\{\mathbf{Q}_m,\mathbf{H}_{m-1},
\gamma_{m-1},Z_{m-1},\mathbf{H}_{m}\}$ to be the {\em global system
state} at the $m$-th slot and
$\chi_{k,m}=\{Q_{k,m},H_{k,m-1},\gamma_{m-1},Z_{m-1},H_{k,m}\}$ to be
the {\em local system state} which is observable locally at the
$k$-th user. Note that $\{\gamma_{m-1},Z_{m-1}\}$ is the common
information for all users, and $\{Q_{k,m},H_{k,m-1},H_{k,m}\}$ is the
local information for the $k$-th user. Given the observed local
system state realization $\chi_{k,m}$, the $k$-th user should adjust
the transmission power according to a {\em stationary power control
policy} $\pi_{P}$, which is formally defined below.
\begin{Def}[Stationary Power Control Policy]\footnote{When {\em
transmission probability} approach is applied, The local system state for the power control
policy should be
$\chi_{k,m}=\{Q_{k,m},H_{k,m-1},\varphi_{m-1}(H),Z_{m-1},H_{k,m}\}$. We further discretize the transmission probability function
$\varphi(H)$ to make the system state discrete. The optimization of the control policy is the similar solution path as the threshold approach.}
The stationary power control policy for single user $\pi_{P}:
\mathcal{N}\times
\mathcal{S}\times\mathcal{S}\times\mathcal{Z}\times\mathcal{S}\rightarrow
\mathbb{R}$ is defined as the mapping from current local system
state for $k$-th user, to current slot's transmit power
$\pi_{P}(\chi_{k,m})=P_{k,m}$. The set of all feasible stationary
policies $\pi_{P}$ is defined as $\mathcal{P}_{P}=\{\pi_{P}:
\pi_{P}(\chi_{k,m})\geq 0\}$. Note that $P_{k,m}=0$ for all
$H_{k,m}<\gamma_{m}$, because current slot's CSI is lower than the
threshold.
\end{Def}

For simplicity, let $\pi=\left\{\pi_{\gamma},\pi_{P}\right\}$ denote
the joint control policy of all the $K$ users. The corresponding set
of stationary joint control policy is given by
$\mathcal{P}=\{\mathcal {P}_{\gamma},\mathcal{P}_{P}\}$ . As a
result,
$\pi(\boldsymbol{\chi}_m)=\{\pi_{\gamma}(\gamma_{m-1},Z_{m-1}),\{\pi_{P}(\chi_{k,m})\}_{k=1}^K\}=\{\gamma_m,\{P_{k,m}\}_{k=1}^K\}$.

In practice, the user with empty buffer will not transmit even if
its local CSI exceeds the system threshold, and this is one
important technical challenge in the delay analysis of S-ALOHA
network. Instead of dealing with the delay for the original S-ALOHA
network, we shall utilize the technique of {\em dominant system}
\cite{Dominant:System} to obtain an upper bound of the delay
performance. In the dominant system, we assume users always have
{\em virtual packets} to send (even if the buffer is empty) and
therefore, the delay performance associated with the dominant system
is always an upper bound of the actual system. Yet, the bound is
asymptotically tight in the large delay regime.

\section{Problem Formulation}\label{sec:problem}
In this section, we shall first formulate the delay-optimal control
policy problem, and then formally introduce DEC-MDP model. We show
that our problem belongs to the memoryless policy case of {\em
DEC-MDP} in which finding the optimal policy is computationally intractable.

\subsection{System Delay}
Due to the nature of random access, the queues of the $K$ users are
coupled together via the control policy. When the system threshold
is small, there will be a high probability of having more than one
users sending packet, leading to collision and wastage of power
resource. On the other hand, when the system threshold is high,
there is non-negligible probability of having no user sending
packet, leading to wastage of idle time. Similarly, individual user
may want to increase the transmit power when the local CSI is good
but if there is collision, the transmitted power is wasted. In this
paper, we seek to find an optimal stationary control policy to
minimize the average delays of the $K$ competing users subject to
average transmit power constraint for single user. Specifically,
the average delay for the $k$-th user is
\begin{equation}
\label{eq:avg_delay} \overline{T_k}(\pi) = \limsup_{M}
\frac{1}{M}\mathbb{E}\left[\sum\nolimits_{m=1}^MQ_{k,m}\right]
\forall k\in \{1,...,K\}
\end{equation}
and average transmit power constraint is given by:
\begin{equation}
\label{eq:pwr_con}
\overline{P_{k}}(\pi)=\limsup_{M}\frac{1}{M}\mathbb{E}\left[\sum\nolimits_{m=1}^MP_{k,m}\right]
\leq P_{0}
\end{equation}
where $P_{k,m}$ is the transmitted power determined by
$\pi(\chi_{k,m})$, and $P_{0}$ is the average power constraint for
single user. The delay-optimal control problem can be formally
written as:
\begin{Prob}[Delay Optimal S-ALOHA Control Policy]
\label{prob:system} Find a stationary control policy $\pi$ that
minimizes
\begin{eqnarray}
\label{eq:asymm}
J^{\pi}(\chi_{1})&=&\sum\nolimits_{k}\left[\overline{T_k}(\pi)
+\xi\overline{P_{k}}\right]\\
&=&\limsup_{M}\frac{1}{M}\sum\nolimits_{m,k}\mathbb{E} \left[
g_k(\chi_{k,m},\pi(\chi_{k,m})) \right]\nonumber
\end{eqnarray} where $g_k(\chi_{k,m},\pi(\chi_{k,m}))=Q_{k,m} + \xi
P_{k,m}$ is the per-stage system price \footnote{In
\cite{Puterman:2005}, it is named {\em price}, yet called {\em cost}
in \cite{Bertsekas_V1:1995}. If it is called a {\em reward}, then
the problem is to maximize the reward.} function and $\xi>0$ is the
Lagrange multipliers corresponding to the average power constraints
in (\ref{eq:pwr_con}).
\end{Prob}

\subsection{DEC-MDP Model}
\label{sec:dec-mdp} Problem \ref{prob:system} in (\ref{eq:asymm}) in
fact belongs to the class of infinite horizon DEC-MDP, which is
formally defined below \cite{DEC-MDP:2002}:
\begin{Def}[DEC-MDP]
An $K$-agent DEC-MDP is given as a tuple
\begin{equation}
\{I,S,A,P(s'|s,a),R(s,a),p_0\}\nonumber
\end{equation}
where $I=\{1,..,K\}$ is a set of agents, $S=\{S_k\}$ is a finite set
of states, $A=\{A_k\}$ is a set of joint actions, $S_k$ and $A_k$ is
available to agent $k$, $P(s'|s,a)$ is the transition probability
that transits from state $s$ to $s'$ given joint action $a$ taken,
$R(s,a)$ is the price function given in state $s$ and joint action
$a$ taken, $p_0$ is the initial state distribution of the system
\footnote{More details about the infinite horizon DEC-MDP is provided
in \cite{DEC-MDP:PI} and the references therein.}.
\end{Def}

The association between Problem \ref{prob:system} and DEC-MDP is as
follows: We have $s_k=\chi_{k,m}$, $a_k=\pi$, $P(s'|s,a)$ can be
easily obtained from local system state transition $P(s_k'|s_k,a_k)$
given in lemma \ref{lem:state_tran}, and
$R(s,a)=\sum\nolimits_{k=1}^K\left[
g_k(\chi_{k,m},\pi_k(\chi_{k,m})) \right]$.

When the policy is given by a mapping from histories of local system
state $\{s_{k,1},...s_{k,m},...\}$ to actions $a_k\in A_k$, the
problem is {\em undecidable}\footnote{Undecidability is a formal
term in the computational complexity theory used to address the
computability and complexity issue on decision problems. A decision
problem is called (recursively) undecidable if no algorithm can
decide it, such as for Turing¡¯s halting problem. It has nothing to
do with whether an optimal solution of an optimization problem exist
or not (or have multiple solutions), because that depends
fundamentally on the structure of the problem. Yet, even if an
optimal solution of an undecidable problem exists theoretically,
there is no algorithm (iterative) to obtain the optimal solution and
terminates \cite{undecidability:cornerstone}.} \cite{Madani:POMDP}.
When the policy is given by a mapping from current local system
state $s_k$ to actions $a_k\in A_k$, it is called memoryless or
reactive policy. In that case, the problem is {\em
NP-hard}\cite{Littman:1994,Meuleau:1999}. As a result, it is very
difficult to obtain the optimal solution for the Problem
\ref{prob:system}. Instead of brute-force solution, we shall try to
exploit the special structure of our problem to obtain low
complexity solutions.

\section{Delay-Optimal Control Problem in Symmetric Network}\label{sec:symm}
In this section, we will focus on exploiting the special structure
of the symmetric network. We shall first solve an optimal power
control policy by a reduced state MDP for any given threshold
control policy. To solve the threshold control problem, we utilize
the collision channel mechanism and derive a low complexity
solution.

\subsection{Embedded Markov Chain under a Given Threshold Control Policy}
For a given threshold control policy, the observed local system
state for single user is actually evolved as a Markov chain.
Specifically, the transition probability conditioned on the power
control policy $\pi_{P}$ is given in the following lemma.

\begin{Lem}[Transition Probability of Local System State]
\label{lem:state_tran}
 At $m$-th slot, the current state of the $k$-th user is
$\chi_{k,m}=\{Q_{k,m},H_{k,m-1},\gamma_{m-1},Z_{m-1},H_{k,m}\}$.
Conditioned on $\pi_{P}$, the transition probability to the next
slot is given by:
\begin{equation}
\label{eq:tran_local}
\begin{array}{l}
\Pr\{\chi_{k,m+1}|\chi_{k,m},\pi_{P}(\chi_{k,m})\}=\mathbb{I}\left(\gamma_m=\pi_{\gamma}(\gamma_{m-1},Z_{m-1})\right)\\
\times\Pr\{H_{k,m+1}|H_{k,m}\}\Pr\{Z_{m}|Z_{m-1},\{H_{k,i},\gamma_{i}\}_{i=m-1}^m\}\\
\times\Pr\{Q_{k,m+1}|\chi_{k,m},Z_{m},\pi_{P}(\chi_{k,m})\}\end{array}
\end{equation}
where $\mathbb{I}(X)$ is an indicate function, which is equal to 1
when event $X$ is true and 0 otherwise.
\end{Lem}
\begin{proof}
Please refer to appendix \ref{app:tran}.
\end{proof}

\subsection{Reduced State MDP Formulation and Optimal Power Control Policy}
For a given threshold control policy in (\ref{eq:asymm}), we seek to
find an optimal power control policy to minimize
\begin{equation}
\label{eq:fix_threshold}
J^{\pi_{P}}(\chi_{1})=\lim_{M}\frac{1}{M}\sum\nolimits_{k,m}\mathbb{E}
\left[ g(\chi_{k,m},\pi_P(\chi_{k,m})) \right]
\end{equation}
Note that, power control policy is a function of local system state,
and for the $k$-th user, its local system state transition
probability is given in (\ref{eq:tran_local}). The optimal power
control policy in (\ref{eq:fix_threshold}) could be decoupled into
$K$ single-user optimization problems, which can be modeled as a MDP
and summarized as following lemma.
\begin{Lem}[Power Control Optimization for Single User]
\label{lem:pwr} The optimal power control policy\footnote{The power
action set is compact, due to finite transmit power in practice. By
Theorem 8.4.7 in \cite{Puterman:2005}, there exists a stationary and
deterministic policy that is average optimal. Thus, it is no loss of
optimality for this power control policy.} minimizing the whole
system delay can be modeled as a single user MDP problem, with state
space given by local system state $\chi_{m}$ (ignoring user index
$k$). The transition probability is given by
$\Pr\{\chi_{m+1}|\chi_{m},\pi_{P}(\chi_{m})\}$ from lemma
\ref{lem:state_tran}, and average price is given by:
\begin{equation}
J^{\pi_{P}}(\chi_{1})=\lim_{M}\frac{1}{M}\sum\nolimits_{m}\mathbb{E}
\left[ g(\chi_{m},\pi_P(\chi_{m})) \right]
\end{equation}
\end{Lem}

For the infinite horizon MDP, the optimal policy can be obtained by
solving the {\em bellman equation} recursively w.r.t
$(\theta,\{V(\chi)\})$ as below:
\begin{eqnarray}
\label{eq:bellman}
&&V(\chi_m)+\theta=\inf_{a(\chi_m)}\biggl\{g(\chi_m,a(\chi_m))+\nonumber\\
&&\quad\quad\sum_{\chi_{m+1}}\Pr\{\chi_{m+1}|\chi_m,a(\chi_m)\}V(\chi_{m+1})\biggr\}
\end{eqnarray}
where $a(\chi_m)=\pi_P(\chi_m)$ is the power allocation when state
is $\chi_m$. If there is a $(\theta,\{V(\chi)\})$ satisfying
(\ref{eq:bellman}), then $\theta$ is the optimal average price per
stage $J^{\pi_{P}}(\chi_{1})$ and the corresponding optimizing
policy is given by $a^*(\chi_m)$, the optimizing action of
(\ref{eq:bellman}) at state $\chi_m$.

Value or policy iteration can be used to solve the bellman equation
(\ref{eq:bellman}) \cite{Puterman:2005,Bertsekas_V1:1995}. The challenge of the two
iteration algorithm lies in the size of the local state space. To
reduce the complexity, we shall recast the original MDP in lemma
\ref{lem:pwr} into a {\em reduced state MDP}. Let's partition the
policy $\pi_P$ into a collection of actions, the above MDP could be
further reduced to a simpler MDP over a {\em reduced state}
$\hat{\chi}_m=\{Q_{m},H_{m-1},\gamma_{m-1},Z_{m-1}\}$
only\footnote{A similar technique was also used in
\cite{Vince:delay,Vince:reduced}}. Specifically, we have following
definition:
\begin{Def}[Conditional Action]
\label{def:con_action} Given a policy $\pi_P$, we define
$\pi_\mathbf{P}(\hat{\chi}_m)=\{\pi_P(\chi_m):\chi_m=(\hat{\chi}_m,H_m)\forall
H_m\}$ as the collection of actions under a given reduced state
$\hat{\chi}_m$ for all possible current slot's CSI $H_m$. The policy
$\pi_P$ is therefore equal to the union of all conditional actions,
i.e., $\pi_P=\bigcup_{\hat{\chi}}\pi_\mathbf{P}(\hat{\chi})$.
\end{Def}

Taking conditional expectation (conditioned on $\hat{\chi}$) on both
sides of (\ref{eq:bellman}), and letting
$\widetilde{V}(\hat{\chi}_m)=\mathbb{E}[V(\chi_m)|\hat{\chi}_m]=\sum\limits_{H_m}\Pr\{H_{m}|H_{m-1}\}V(\chi_m)$,
the Bellman equation becomes:
\begin{eqnarray}
\label{eq:red_bellman}
&&\widetilde{V}(\hat{\chi}_m)+\theta =\inf_{a(\chi_m)}\Biggl\{\sum_{H_m}\Pr\{H_{m}|H_{m-1}\}\biggl(g(\chi_m,a(\chi_m))\nonumber\\
&&+\sum_{\chi_{m+1}}\Pr\{\chi_{m+1}|\chi_{m},a(\chi_{m})\}V(\chi_{m+1})\biggr)\Biggr\}\nonumber\\
&&=\inf_{a(\chi_m)}\biggl\{\sum\nolimits_{H_m}\Pr\{H_{m}|H_{m-1}\}g(\chi_m,a(\chi_m))\nonumber\\
&&+\sum_{\hat{\chi}_{m+1}}\sum_{H_m}\Pr\{H_{m}|H_{m-1}\}\Pr\{\hat{\chi}_{m+1}|\chi_{m},a(\chi_m)\}\nonumber\\
&&\times\sum_{H_{m+1}}\Pr\{H_{m+1}|H_{m}\}V(\hat{\chi}_{m+1},H_{m+1})\biggr\}\nonumber\\
&&=\inf_{\mathbf{a}(\hat{\chi}_m)}\biggl\{\widetilde{g}(\hat{\chi}_m,\mathbf{a}(\hat{\chi}_m))+\nonumber\\
&&\sum\nolimits_{\hat{\chi}_{m+1}}\Pr\{\hat{\chi}_{m+1}|\hat{\chi}_{m},\mathbf{a}(\hat{\chi}_m)\}\widetilde{V}(\hat{\chi}_{m+1})\biggr\}
\end{eqnarray}
where $a(\chi_m)=\pi_P(\chi_m)$ is a single power allocation action
at state $\chi_m$ and
$\mathbf{a}(\hat{\chi}_m)=\pi_\mathbf{P}(\hat{\chi}_m)$ is the
collection of power allocation actions under a given reduced state
$\hat{\chi}_m$. Furthermore, $\widetilde{g}(\hat{\chi}_m,
\mathbf{a}(\hat{\chi}_m))$ is the conditional per-stage price
function given by:
\begin{eqnarray}
\label{eq:redu_g}
\widetilde{g}(\hat{\chi}_m,\mathbf{a}(\hat{\chi}_m))&=&\mathbb{E}[g(\hat{\chi}_m,H_m,a(\chi_m))|\hat{\chi}_m]\\
&=&Q_{m}+\xi\left(\sum\nolimits_{H_{m}}\Pr\{H_{m}|H_{m-1}\}P_m\right)\nonumber
\end{eqnarray}

As a result, the original MDP is equivalent to a reduced state MDP,
which is summarized in the following lemma.
\begin{Lem}[Equivalent MDP on a Reduced State Space]
\label{Lem:reduced} The original MDP in lemma \ref{lem:pwr} is
equivalent to the following reduced state MDP with state space given
by $\hat{\chi}_m$, average price given by:
\begin{eqnarray}
J^{\pi_{P}}(\chi_{1})=\limsup_{M}\frac{1}{M}\sum\nolimits_{m=1}^{M}\mathbb{E}
\left[\widetilde{g}(\hat{\chi}_m,\mathbf{a}(\hat{\chi}_m))\right]
\end{eqnarray}
$\Pr\{\hat{\chi}_{m+1}|\hat{\chi}_{m},\mathbf{a}(\hat{\chi}_m)\}$ is
the states transition kernel given by:
\begin{eqnarray}
&&\Pr\{\hat{\chi}_{m+1}|\hat{\chi}_{m},\mathbf{a}(\hat{\chi}_m)\}\\
&&\quad=\sum\nolimits_{H_m}\Pr\{H_{m}|H_{m-1}\}\Pr\{\hat{\chi}_{m+1}|\chi_{m},a(\chi_m)\}\nonumber
\end{eqnarray}
\end{Lem}

The bellman equation for reduced state MDP is given in
(\ref{eq:red_bellman}). Note that while the reduced state MDP is
defined over the partial state $\hat{\chi}$, the power allocation is
still a function of the original complete local system state. In
fact, for realization of the reduced state $\hat{\chi}_m$, the
solution of the reduced MDP gives the conditional actions for
different realization of $H_m$.

\subsection{Delay-Optimal Power Control Solution}
Value or policy iteration can be used to solve the bellman equation
(\ref{eq:red_bellman}), and the convergence of the iteration
algorithms is ensured by the following lemma.
\begin{Lem}[Decidability of the Unichain of Reduced State]
\label{lem:unichain} The unichain \footnote{In \cite{Puterman:2005},
unichain is defined as a single recurrent class plus a possibly
empty set of transient states.} of the reduced state MDP in lemma
\ref{Lem:reduced} is decidable under all power control policy.
\end{Lem}
\begin{proof}
Please refer to appendix \ref{app:unichain}.
\end{proof}

The number of unichains of the reduced state MDP in
(\ref{Lem:reduced}) depends on the number of recurrent classes of
local system state (excluding the queue state $Q$) in
$\hat{\chi}_m$, i.e., $\Phi_{m}=\{H_i,\gamma_i,Z_i\}_{i=m-1}$. The
value or policy iteration could be applied to different unichains
respectively, while the convergence and unique solution is ensured
\cite{Puterman:2005}. Specifically, the bellman equation
(\ref{eq:red_bellman}) could be elaborated in an offline manner as
follows:
\begin{equation}
\label{eq:pwr_pol}
\begin{array}{l}
\widetilde{V}(\hat{\chi}_m)+\theta=
\inf\limits_{\pi_\mathbf{P}(\hat{\chi}_m)}\biggl\{\widetilde{g}(\hat{\chi}_m,\pi_\mathbf{P}(\hat{\chi}_m))+\\
\sum\limits_{\Phi_{m+1}}\Pr\{\Phi_{m+1}|\Phi_{m}\}\Bigl[\tau\lambda\widetilde{V}((q_m+1)_{\bigwedge N},\Phi_{m+1})+\\
\tau\mu\widetilde{V}((q_m-1)^+,\Phi_{m+1})+\left(1-\tau\lambda-\tau\mu\right)\widetilde{V}(q_m,\Phi_{m+1})\Bigr]\biggr\}
\end{array}
\end{equation}
where $\mu=\mu(\chi_m,Z_m,\pi_P(\chi_m))$ is the mean packet service
rate in (\ref{eq:mu}), $x_{\bigwedge N}=\min\{x,N\}$, and let $P(\chi_m)=\pi_P(\chi_m)$. In the right hand side of
(\ref{eq:pwr_pol}), $P(\chi_m)$ only influence $\mu$ and
$\widetilde{g}$ in (\ref{eq:redu_g}). Specifically,
$\Pr\{\Phi_{m+1}|\Phi_{m}\}$ is presented simply as
$\Pr\{H_m|H_{m-1}\}\Pr\{Z_{m}|Z_{m-1}\}$. Hence, the optimal power
control policy for a system state $\chi_m$ is thus given by:
\begin{equation}
\label{eq:pwr_opt}
\begin{array}{l}
P(\chi_m)=\arg\min\limits_{P(\chi_m)}\biggl\{\\
\sum\limits_{H_m=\gamma_{m}}^{S_J}\Pr\{H_m|H_{m-1}\}\Bigl[\xi
P(\chi_m)+\Pr\{Z_{m}=1|Z_{m-1}\}\\
\frac{W\tau}{\overline{N_b}}\log_{2}(1+\frac{P(\chi_m)H_{m}}{N_0W})\delta(q_m,H_m,\gamma_m)\Bigr]\biggr\}\\
=\biggl(-W\tau\Pr\{Z_m=1|Z_{m-1}\}\delta(q_m,H_m,\gamma_m)/\left(\overline{N_b}\xi\ln2\right)\\
\quad -N_0W/H_m\biggr)^+
\end{array}
\end{equation}
where
$\delta(q_m,H_m,\gamma_m)=\widetilde{V}((q_m-1)^+,H_m,\gamma_m,Z_m=1)-\widetilde{V}(q_m,H_m,\gamma_m,Z_m=1)$.
Note that the optimal power control action depends on the local CSI
via the standard water-filling form. On the other hand, it also
depends on the local QSI and common feedback $Z$ through the
water-level\footnote{As a sanity check, when the CSI are i.i.d and the the control policies
are not function of QSI (i.e., ($\pi_{\gamma}(H):
\mathcal{S}\rightarrow \mathcal{S},\pi_{P}(H):
\mathcal{S}\rightarrow \mathbb{R}$)), using similar reduced state MDP technique,
the optimal power control policy is represented as:
$(-W\tau(\sum_{S_i<\gamma}\pi_i)^{K-1}/(\overline{N_b}\widetilde{\xi}\ln2)-N_0W/H_m)^+$. Where
$\widetilde{\xi}=(\widetilde{V}(q_m)-\widetilde{V}((q_m-1)^+))/\xi$
is the new Lagrange Multiplier, and considered as a constant since
the QSI influence is ignored.
Then optimal threshold $\gamma$ can be obtained. It is the same as the binary
scheduling with power control w.r.t the CSI studied in the \cite{Thrput:pwr} called {\em Variable-Rate Algorithms}. }. Using the optimal power allocation policy, the
transition probability of reduced state is
$\Pr\{\hat{\chi}_{m+1}|\hat{\chi}_{m}\}=\Pr\{Q_{m+1}|\chi_m,Z_m,\pi_P(\chi_m)\}\Pr\{\Phi_{m+1}|\Phi_{m}\}$.
The stationary distribution of $\hat{\chi}$, denoted
$\omega(\hat{\chi})$, could be found by the linear equations
$\omega(\hat{\chi}_j)=\sum_i\omega(\hat{\chi}_i)\Pr\{\hat{\chi}_{j}|\hat{\chi}_{i}\}$.
Finally, the Lagrange multiplier $\xi$ is chosen to satisfy the
average power constraint per user $P_0$:
\begin{equation}
\label{eq:avg_pwr}
P_0=\omega(\hat{\chi}_m)\sum_{H_m}\Pr\{H_m|H_{m-1}\}P(\chi_m)
\end{equation}

\subsection{Threshold Control Policy}
Threshold control policy is determined based on the common
information $\{\gamma_{m-1},Z_{m-1}\}$. The full exploitation of the
known information is critical to improve the delay performance of
the system. In fact, the common information
$\{\gamma_{m-1},Z_{m-1}\}$ could be used to exploit the memory of
all the $K$ competing users' fading channels, and predict their
transmission events at the current slot. Specifically, in the
collision channel, data will be successfully received by the AP in
the S-ALOHA network, if and only if exactly one user transmits at
one slot. Consequently, the known information shall be chosen to
ensure the user with the largest CSI will transmit alone with the
highest probability. Based on this observation, we propose a {\em
larger CSI higher priority} (LCSIHP) threshold control policy as
follows:
\begin{eqnarray}
\label{eq:gamma_symm} \gamma_m^*&=&
\pi_{\gamma}(\gamma_{m-1},Z_{m-1})\\
&=&\arg\max_{\gamma_m}\Pr\{\text{only 1 user
transmits}|\gamma_{m-1},Z_{m-1}\}\nonumber
\end{eqnarray} where $\Pr\{\text{only 1 user transmits}|\gamma_{m-1},Z_{m-1}\}$ is
given by:
\begin{equation}
\begin{array}{ll}
&\Pr\{\text{only 1 user transmits}|\gamma_{m-1},Z_{m-1}\}\\
&=\left\{
\begin{array}{ll}
K\upsilon\left(\overline{\upsilon}\right)^{(K-1)} & \textrm{if $Z_{m-1}=0$}\\
\Bigl[\zeta\overline{\upsilon}^{(K-1)}+(K-1)\overline{\zeta}\upsilon\overline{\upsilon}^{(K-2)}\Bigr] & \textrm{if $Z_{m-1}=1$}\\
\begin{array}{l}\sum_{k=2}^{K}\biggl[p_{\gamma_{m-1},2}^{(K,k)}\Bigl(k\zeta\overline{\zeta}^{k-1}\overline{\upsilon}^{(K-k)}+\\
(K-k)\overline{\zeta}^{k}\upsilon\overline{\upsilon}^{(K-k-1)}\Bigr)\biggr]\end{array}
& \textrm{if $Z_{m-1}=e$}
\end{array} \right.\end{array}
\end{equation}

where $p_{\gamma,2}^{(K,k)}$ given in (\ref{eq:pr_e}), is a function
of $\gamma$, and
$\{\upsilon,\overline{\upsilon},\zeta,\overline{\zeta}\}$ given in
(\ref{eq:tx_event}) (ignoring user index $k$) are functions of
$\{\gamma_i\}_{i=m-1}^m$. The way to obtain (\ref{eq:gamma_symm}) is
to treat the previous slot's transmitted and non-transmitted users
separately. As a sanity check, note that when the CSI are i.i.d,
$\zeta=\upsilon$ and $\overline{\zeta}=\overline{\upsilon}$ for all
$\{\gamma_{m-1},Z_{m-1}\}$, equation (\ref{eq:gamma_symm}) is
reduced to
$\gamma_m^*=\arg\max_{\gamma_m}K\upsilon\left(\overline{\upsilon}\right)^{(K-1)}$.
$\gamma_m^*$ is the same for all the slot to maximize the
probability that only one user will transmit.

\subsection{Summary of the Solution in Symmetric Network}
\label{sec:sym_summary} The overall power and threshold control
solution in symmetric network consists of an offline procedure and
an online procedure and they are summarized below.
\[
\framebox{\parbox{8.5cm}{ Offline Procedure: The output of the
offline procedure is optimal power allocation $\pi_P(\chi)$, which
will be stored in a table and used in the online procedure.
\begin{itemize}
\item{\bf Step 1) Determination of the threshold control
policy:} Figure out the threshold control policy from
(\ref{eq:gamma_symm}) for different realization of
$\{\gamma_{m-1},Z_{m-1}\}$.
\item{\bf Step 2) Acquire unichains of reduced state:}
From the given threshold control policy, obtain the recurrent
classes of the reduced state $\hat{\chi}$ from lemma
\ref{lem:unichain}.
\item{\bf Step 3) Determination of the optimal power control
policy:} For a given $\xi$, determine $\theta(\xi)$,
$\{\widetilde{V}(Q_m,H_{m-1},\gamma_{m-1},Z_{m-1};\xi)\}$ of the
bellman equation (\ref{eq:pwr_pol}) in every unichain of reduced
state by policy or value iteration algorithm. The optimal power
control policy $\pi_P(\chi_m;\xi)$ is then determined in
(\ref{eq:pwr_opt}).
\item{\bf Step 4) Transmit power constraint:} For a given $\xi$, the
average transmit power $P_0$ can be obtained in (\ref{eq:avg_pwr}).
On the other hand, we could use root-finding numerical algorithm to
determine $\xi$ that satisfies a given $P_0$.
\end{itemize}
}}
\]
\[
\framebox{\parbox{8.5cm}{ Online procedure: The homogeneous users
observe $\chi_{m}=\{Q_{m},H_{m-1},\gamma_{m-1},Z_{m-1},H_{m}\}$, the
local system state realization at the beginning of the $m$-th slot
and transmits at a power given by $\pi_P(\chi_m)$. If
$H_{m}<\pi_{\gamma}(\gamma_{m-1},Z_{m-1})$, $P_m=\pi_P(\chi_m)=0$,
i.e., the user will not transmit. }}
\]

The complexity of the online procedure is negligible because it is
simply a table looking up. The complexity of the offline procedure
depends mostly on the solution of power control policy, which
contains an iteration algorithm to solve the bellman equation in
(\ref{eq:pwr_pol}). Specifically, the complexity of the reduced
state MDP is given in following theorem.
\begin{Thm}[Complexity of the Reduced State
MDP]\label{thm:complexity} The worst case complexity of the reduced
state MDP is $O(f(K))$, where $f(K)$ is a monotonic decreasing
function of number of users $K$. Furthermore, there exists a
constant $K_0>0$ such that for all $K>K_0$, the complexity is
reduced to $O(NJ)$.
\end{Thm}
\begin{proof}
Please refer to Appendix \ref{app:complexity}.
\end{proof}

Theorem \ref{thm:complexity} implies that when $K$ is large enough,
there is no need to exploit the memory of the fading channels. The
threshold is fixed to $S_J$ regardless of the common feedback. This
is reasonable because the more competing users we have, the smaller
the chance for single user to transmit. Hence, for sufficiently
large $K$, the users are only allowed to transmit when local CSI
reaches the largest state $S_J$, so as to reduce the intensive
collision. Note that, the complexity of the offline procedure is
substantially reduced, compared to the complexity $O(NJ^3)$ of the
brute-force solution in the original MDP in lemma \ref{lem:pwr}.

\section{Extension to Asymmetric Network}\label{sec:asymmetric}
In this section, we shall extend the delay control framework to
asymmetric S-ALOHA network, in which heterogenous users have
different fading channels. Specifically, let
$\mathcal{S}_k=\{S_{i}\}_{i=1}^{J_k}$ denote a set of $J_k$ CSI
states, $p_{i,j}^k$ denote the state transition probability and
$\pi_j^k$ denote the stationary probability for user $k$. The common
threshold for all users is not applicable for the heterogenous users
and hence, the system threshold $\gamma_m$ is extended to
$\Gamma_m=\{\gamma_{k,m}\}_{k=1}^K$, where $\gamma_{k,m}$ is the
threshold for user $k$. As a result, the threshold control policy is
extended to $\Gamma_{m}=\pi_{\Gamma}(\Gamma_{m-1},Z_{m-1})$, and
power control policy for user $k$ is denoted as
$\pi_{P_k}(\chi_{k,m})$. The set of joint control policy
$\pi=\{\pi_{\Gamma},\{\pi_{P_k}\}_{k=1}^K\}$ is easily redefined as
in section \ref{sec:model}.

\subsection{Optimal Power Control Policy under a Given Threshold Control Policy}
For a given threshold control policy, Lemma \ref{lem:state_tran}
still holds. Due to the extension of single threshold $\gamma_m$ to
system threshold $\Gamma_m$, the transition probability of local
system state of the $k$-th user should be rewritten as:
\begin{eqnarray}
\label{eq:tran_local_asymm}
&&\Pr\{\chi_{k,m+1}|\chi_{k,m},\pi_{P_k}(\chi_{k,m})\}=\\
&&\mathbb{I}\left(\Gamma_m=\pi_{\Gamma}(\Gamma_{m-1},Z_{m-1})\right)\Pr\{H_{k,m+1}|H_{k,m}\}\nonumber\\
&&\times\Pr\{Z_{m}|Z_{m-1},\{H_{k,i},\Gamma_{i}\}_{i=m-1}^m\}\nonumber\\
&&\times\Pr\{Q_{k,m+1}|\chi_{k,m},Z_{m},\pi_{P_k}(\chi_{k,m})\}\nonumber
\end{eqnarray}
where the transition probability of the feedback state $Z$ is not as
simple as the symmetric case shown in appendix \ref{app:feedback}.
For instance, the memory of channel fading of other $(K-1)$ users
should also be exploited through the known information
$\Psi_{k,m-1}=\{H_{k,m-1},\Gamma_{m-1},Z_{m-1}\}$ of user $k$.
Hence, the joint probability of CSI for other users at $m$-th slot
is given by:
\begin{eqnarray}
&&\Pr\{\mathbf{H}_{-k,m}|\Psi_{k,m-1}\}=\\
&&\sum_{\mathbf{H}_{-k,m-1}}\Pr\{\mathbf{H}_{-k,m-1}|\Psi_{k,m-1}\}\Biggl(\prod_{i\neq
k}\Pr\{H_{i,m}|H_{i,m-1}\}\Biggr)\nonumber
\end{eqnarray}
where $\mathbf{H}_{-k,m}=\{H_{i,m}\}_{i=1,i\neq k}^K$ is the set of
all users' CSI at the $m$-th slot, excluding the $k$-th one, and
$\Pr\{\mathbf{H}_{-k,m-1}|\Psi_{k,m-1}\}=\prod\limits_{i\neq
k}\Pr\{H_{i,m-1}\}/\sum\limits_{\mathbf{H}_{-k,m-1}}\prod_{i\neq
k}\Pr\{H_{i,m-1}\}$ is the belief of the possible realization of
$\mathbf{H}_{-k,m-1}$ conditioned on $\Psi_{k,m-1}$. Then, the
feedback transition probability is given by:
\begin{eqnarray}
\label{eq:feedback_asymm}
&&\Pr\{Z_{m}|Z_{m-1},\{H_{k,i},\Gamma_{i}\}_{i=m-1}^m\}=\nonumber\\
&&\sum\nolimits_{\mathbf{H}_{-k,m}}\Pr\{\mathbf{H}_{-k,m}|\Psi_{k,m-1},\Psi_{k,m}\}
\end{eqnarray}

As a result, the power control solution $\pi_{P_k}(\chi_{k,m})$ for
the $k$-th user is similar to lemma \ref{Lem:reduced} except that
transition probability of feedback state $Z_m$ is replaced by
(\ref{eq:feedback_asymm}).

\subsection{Threshold Control Policy}
The system threshold $\Gamma_{m}$ determined by the threshold
control policy will influence the successful transmission
probability of each user. Specifically, let $\alpha_{k,m}(\Gamma_m)$
represent the probability that user $k$ transmits alone at $m$-th
slot, i.e.,
$\alpha_{k,m}=\Pr\left\{H_{k,m}\geq\gamma_{k,m},\bigcup_{i\neq k}
{H_{i,m}<\gamma_{i,m}}\right\}$. Note that an increase in
$\alpha_{k,m}$ for user $k$ will result in a decrease in
$\alpha_{i,m}$ for all $i\neq k$. Hence, there is a tradeoff
relationship among the probability of successful transmission of the
$K$ users. Unlike the symmetric case, the threshold control policy
shall not only improve the delay performance, but also consider the
fairness among the $K$ heterogenous users. In S-ALOHA network
\cite{Yu:2006} and centralized system \cite{Viswanath:2002}, the
authors proposed the product-optimization form to take the fairness
into consideration. Similarly, we consider a system threshold
control policy that maximizes the product-probability:
\begin{eqnarray}
\label{eq:gamma_asymm}
\Gamma_{m}^*&=&\arg\max\nolimits_{\Gamma_{m}}\prod\nolimits_k\alpha_{k,m}
\end{eqnarray}

The product-maximization could prevent users from having very low
successful transmission probability. Similar to the symmetric case,
we shall exploit the common information $\{\Gamma_{m-1},Z_{m-1}\}$
to enhance the probability of successful transmission of $K$
competing users over a collision channel. Given all the transmission
event $\{B_{i,m-1}\}_{i=1}^K$ (defined in definition
\ref{def:transmission}) at the previous slot, the probability that
user $k$ transmits alone at current slot is given by:
\begin{eqnarray}
\label{eq:transmit_alone}
&&\alpha_{k,m}(\Gamma_{m},\Gamma_{m-1},\{B_{i,m-1}\}_{i=1}^K)=\nonumber\\
&&\Pr\{A_{k,m}|\gamma_{k,m},\gamma_{k,m-1},B_{k,m-1}\}\nonumber\\
&&\prod\nolimits_{i\neq
k}\Pr\{\overline{A}_{i,m}|\gamma_{i,m},\gamma_{i,m-1},B_{i,m-1}\}
\end{eqnarray}
Substituting into (\ref{eq:gamma_asymm}), $\Gamma_{m}^*$ could be
decoupled into single user optimization problem, i.e.,
\begin{eqnarray}
&&\gamma_{k,m}^*=\arg\max_{\gamma_{k,m}}\Pr\{A_{k,m}|\gamma_{k,m},\gamma_{k,m-1},B_{k,m-1}\}\nonumber\\
&&\left(\Pr\{\overline{A}_{k,m}|\gamma_{k,m},\gamma_{k,m-1},B_{k,m-1}\}\right)^{K-1}
\end{eqnarray}

From $\{\Gamma_{m-1},Z_{m-1}\}$, we can calculate the probability
that any specific user transmitted before and hence, the threshold
control policy can be solved by single user optimization problem
given by:
\begin{equation}
\label{eq:gamma_asymm_single}
\gamma_{k,m}^*=\\
\left\{\begin{array}{ll}
\arg\max\limits_{\gamma_{k,m}}\upsilon_k\left(1-\upsilon_k\right)^{K-1}&\textrm{if
$Z_{m-1}=0$}\\
\arg\max\limits_{\gamma_{k,m}}\begin{array}{l}\overline{\rho}_k\upsilon_k\left(1-\upsilon_k\right)^{K-1}\\
+\rho_k\zeta_k\left(1-\zeta_k\right)^{K-1} \end{array} &\textrm{if
$Z_{m-1}\neq 1$}
\end{array}\right.
\end{equation}
where $\{\upsilon_k,\zeta_k\}$ are obtained in (\ref{eq:tx_event}),
$\rho_k=\Pr\{A_{k,m-1}|\Gamma_{m-1},Z_{m-1}\}$ is the conditional
probability that user $k$ transmits at the $(m-1)$-th slot, and
$\overline{\rho}_k=(1-\rho_k)$. Hence, we have
\begin{equation}
\rho_k=\left\{\begin{array}{ll} \eta_k\prod_{i\neq
k}\overline{\eta}_i/\sum_{k}\eta_k\prod_{i\neq
k}\overline{\eta}_i&\textrm{if
$Z_{m-1}=1$}\\
\frac{\eta_k\left(1-\prod_{i\neq
k}\overline{\eta}_i\right)}{\left(1-\prod_{i}\overline{\eta}_i-\sum_j\prod_{i\neq
j}\eta_j\overline{\eta}_i\right)}&\textrm{if
$Z_{m-1}=e$}\end{array}\right.
\end{equation}
where
$\eta_k=\Pr\{A_{k,m-1}|\gamma_{k,m-1}\}=\sum_{S_j\geq\gamma_{k,m-1}}\pi_{j}^k$
is the transmission probability of user $k$, given the threshold is
$\gamma_{k,m-1}$, and $\overline{\eta}_k=1-\eta_k$.

\subsection{Summary of the Solution in Asymmetric Network}
The overall solution of the control policy in asymmetric network
also consists of an offline procedure and an online procedure.
Compared with the symmetric case, the optimal power control policy
$\pi_{P_k}(\chi)$ is not the same for all the heterogenous users. In
the offline procedure, $\pi_{P_k}(\chi)$ should be calculated and
stored in corresponding user's table for online looking up.

Similarly, the online procedure is a table looking up and hence, the
complexity is negligible. Since the threshold control policy is
decoupled to one dimensional optimization problem for single user,
the complexity of the offline procedure still depends mostly on the
iteration algorithm. Due to the extension of the system threshold,
the number of reduced state $\hat{\chi}_{k,m}$ is
$\mathcal{O}(N\prod_kJ_k)$. However, theorem \ref{thm:complexity}
still holds in the asymmetric network. For sufficiently large $K$,
the threshold control policy will increase the threshold of each
user so as to avoid intensive collision. As a result, the number of
possible $\Gamma$ states is substantially reduced and the asymptotic
complexity of user $k$ becomes $O(NJ_k)$ as in the symmetric case.

\section{Numerical Results and Discussions}\label{sec:simulation}
In this section, we shall illustrate the delay performance of the
proposed control policy via numerical simulations. We set the time
of a slot $\tau=1$ms, bandwidth $W=1$KHz. We model the packet
arrival and CSI event follows the assumption in the system model
(Section \ref{sec:model}). With different simulation scenarios, we
calculate the optimal policies in offline. In the online
application, the users simply implement the policy at each slot
corresponding to the system state observed in that slot. The packet
will stay in the buffer until it is successfully serviced, and the performance is evaluated with sufficient
realizations.

\begin{table*}[t]
\begin{center}
\caption{Two FSMC CSI Models with the Same States yet Different
Transition Probability (User1/User2)}
\begin{tabular}{|c|c|c|c|c|c|c|c|c|c|c|}
\hline & $H_1$& $H_2$& $H_3$& $H_4$& $H_5$& $H_6$& $H_7$& $H_8$& $H_9$& $H_{10}$\\
\hline States & 0.055& 0.074& 0.112& 0.153& 0.237& 0.531& 0.894& 1.343& 2.588& 4.493 \\
\hline $H_1$& 0.2/0.25& 0.8/0.75& 0& 0& 0& 0& 0& 0& 0& 0\\
\hline $H_2$& 0.2/0.25& 0.3/0.3& 0.5/0.45& 0& 0& 0& 0& 0& 0& 0\\
\hline $H_3$& 0&  0.25/0.3& 0.35/0.35& 0.4/0.35& 0& 0& 0& 0& 0& 0\\
\hline $H_4$& 0& 0& 0.3/0.34& 0.3/0.3& 0.4/0.36& 0& 0& 0& 0& 0\\
\hline $H_5$& 0& 0& 0& 0.33/0.37& 0.34/0.34& 0.33/0.29& 0& 0& 0& 0\\
\hline $H_6$& 0& 0& 0& 0& 0.33/0.37& 0.34/0.34& 0.33/0.29& 0& 0& 0\\
\hline $H_7$& 0& 0& 0& 0& 0& 0.4/0.36& 0.3/0.3& 0.3/0.34& 0& 0\\
\hline $H_8$& 0& 0& 0& 0& 0& 0& 0.4/0.35& 0.35/0.35& 0.25/0.3& 0\\
\hline $H_9$& 0& 0& 0& 0& 0& 0& 0& 0.5/0.45& 0.3/0.3& 0.2/0.25\\
\hline $H_{10}$& 0& 0& 0& 0& 0& 0& 0& 0& 0.8/0.75& 0.2/0.25\\
\hline $\pi_j^1$& 0.0137& 0.0548& 0.1097& 0.1463& 0.1755& 0.1755& 0.1463& 0.1097& 0.0548& 0.0137\\
\hline $\pi_j^2$& 0.0342& 0.1027& 0.154& 0.1586& 0.1529& 0.1201& 0.0979& 0.0951& 0.0634& 0.0211\\
\hline
\end{tabular}
\label{tab:CSI}
\end{center}
\end{table*}

\begin{figure}
\begin{center}
  {\resizebox{8cm}{!}{\includegraphics{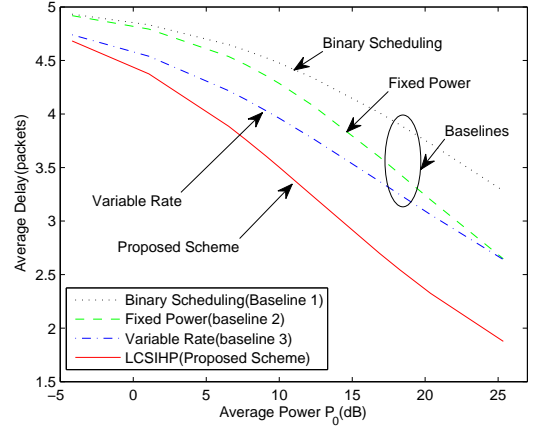}}}
  \end{center}
    \caption{Comparison of the delay performance between proposed control policy and three baselines in symmetric
    network, with $1^{\mathrm{st}}$ CSI model in Table \ref{tab:CSI} for all the homogeneous users. We assume that the buffer length
$N=5$, packet arrival rate $\lambda=1$ for all $K=5$ users, with
mean packet size $\overline{N}_b=1$K bits.}
    \label{fig:sym_delay}
\end{figure}

\begin{figure}
\begin{center}
  {\resizebox{8cm}{!}{\includegraphics{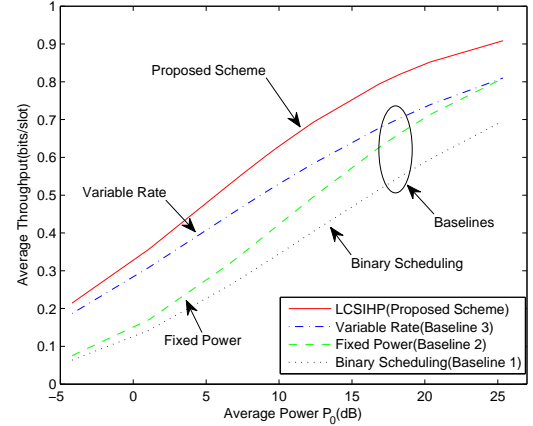}}}
  \end{center}
    \caption{Comparison of the throughput performance between proposed control policy and three baselines in symmetric
    network. The configuration is the same as Fig.\ref{fig:sym_delay}.}
    \label{fig:sym_thrput}
\end{figure}

\begin{figure}
\begin{center}
  {\resizebox{8cm}{!}{\includegraphics{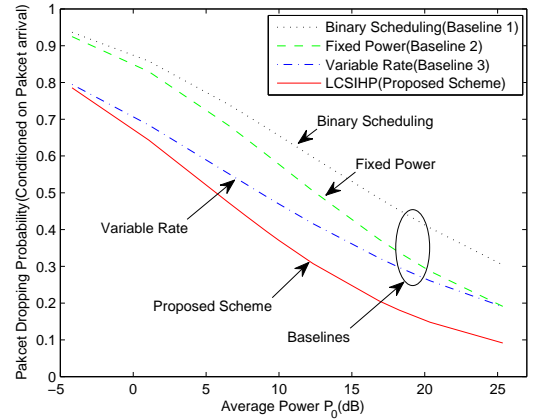}}}
  \end{center}
    \caption{Comparison of Packet Dropping Probability (Conditioned on
 Packet Arrival) between proposed control policy and three baselines
in symmetric network. The configuration is the same as
Fig.\ref{fig:sym_delay}.}
    \label{fig:sym_pr_full}
\end{figure}

Fig.\ref{fig:sym_delay}-Fig.\ref{fig:sym_pr_full} compares the
LCSIHP threshold control policy (corresponding optimal power control
policy) in symmetric network with three reference baselines.
Baseline 1 corresponds to the binary scheduling algorithm in
\cite{Yu:2006}. Baseline 2 corresponds to the LCSIHP threshold
control policy without power control. Baseline 3 corresponds to the
variable-rate algorithm with power control proposed in
\cite{Thrput:pwr}. We observe that there is a significant gain in
both delay and throughput of the proposed policy over these three
baselines. Fig.\ref{fig:sym_pr_full} compares packet dropping
probability (packet arrives when the buffer is full $Q=N$). It shows
that packet dropping performance is also improved by the proposed
policy. This scenario can also be inferred from the optimal power
control policy, which will potentially put more power on the node
with larger QSI to reduce the delay.

Fig.\ref{fig:Delay_asymm} compares the delay performance in
asymmetric network for two heterogenous users. The mean packet
arrival rate is assumed to be $\lambda=2$. Other settings are the
same as the symmetric case. We compare the performance of the
proposed scheme in section \ref{sec:asymmetric} (denoted {\em
Asymm}) with another baseline scheme designed for heterogeneous
users in \cite{Yu:2006}. Specifically, we consider power control
w.r.t CSI under the binary scheduling scheme in \cite{Yu:2006} to
form a competitive baseline (namely {\em BSP}). Observe there is
discontinuity in the delay performance of BSP, and this is because
in small SNR regime, the system threshold for user 2 is lower than
that of user 1 but they become the same in large SNR regime. Observe
that the proposed scheme has significant performance gain in terms
of fairness or delay performance compared with the BSP baseline.
\begin{figure}
\begin{center}
  {\resizebox{8cm}{!}{\includegraphics{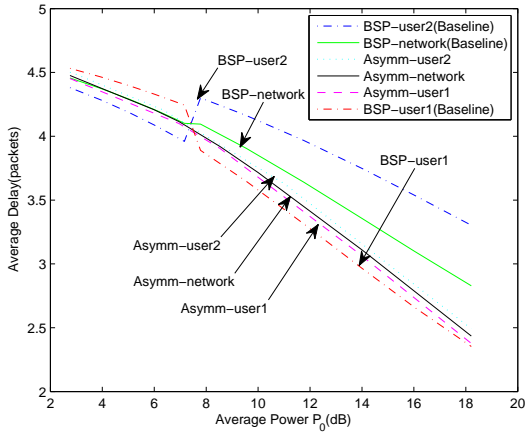}}}
  \end{center}
    \caption{Comparison of delay performance between BSP (power control w.r.t CSI
additionally) and proposed Asymm policy in asymmetric network with
two heterogenous users, and their CSI models are listed in Table
\ref{tab:CSI} (user1/user2). Specifically, BSP-user2 denotes the
delay performance for user 2 under BSP policy, while BSP-user1 is
denoted for user 1. BSP-network denotes the average delay
performance of the two users under BSP policy. Correspondingly, the
notation started with Asymm denotes the delay performance under
Asymm policy.}
    \label{fig:Delay_asymm}
\end{figure}

Fig.\ref{fig:Delay_asymm_10users} compares the delay performance in
a larger asymmetric network. There are 10 heterogeneous users which
are divided into 5 groups. In each group, there are two homogeneous
users. Furthermore, we assume a larger buffer size $N=10$, and
$\lambda=0.4$. It can be observed that in a larger network, the
fairness improvement is less obvious. This is because the threshold
is increased to avoid the intensive collision both under Asymm or
BSP policy, and the freedom of the improvement for Asymm policy is
reduced. However, the delay performance is obviously guaranteed due
to the additional dimension in QSI for the Asymm policy.

\begin{figure}
\begin{center}
  {\resizebox{8cm}{!}{\includegraphics{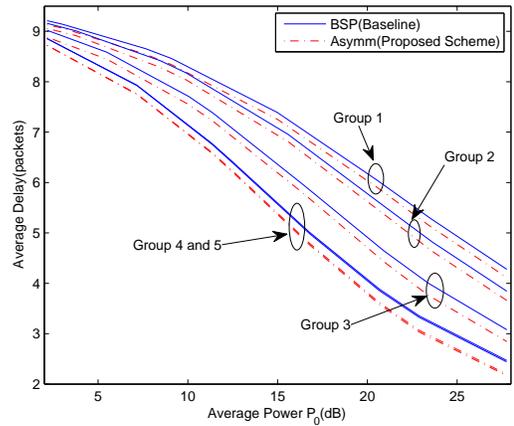}}}
  \end{center}
    \caption{Comparison of delay performance between BSP (power control w.r.t CSI
additionally) and proposed Asymm policy in asymmetric network with
10 heterogenous users. Every group has two homogeneous users.}
    \label{fig:Delay_asymm_10users}
\end{figure}

\begin{figure}
\begin{center}
  {\resizebox{8cm}{!}{\includegraphics{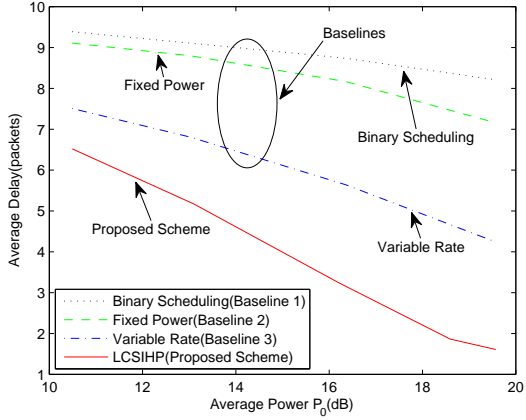}}}
  \end{center}
    \caption{Comparison of the delay performance in 10 users symmetric
    network with capture effect at the AP. Specifically, we consider narrow band transmission ($W=1$KHz).
    The data rate is given by $\widetilde{R}_k=\beta
W\log_2(1+\frac{P_kH_k}{N_0W})$, where $\beta=0.9$, buffer length
$N=10$, and packet arrival rate $\lambda=0.4$, with mean packet size
$\overline{N}_b=1$K bits. When collision occurs, the packet sent by
the $k$-th user will be successfully detected at the AP if
$\widetilde{R}_k\leq
C_k(\text{collision})=W\log_2(1+\frac{P_kH_k}{\sum_{i\neq
k}P_iH_i+N_0W})$. Otherwise, it will be corrupted. }
    \label{fig:Delay_sym_MPR}
\end{figure}

Fig.\ref{fig:Delay_sym_MPR} compares the delay performance of the
random access channel with capture effect\footnote{In our original
formulation, we have set the transmit data rate according to the
instantaneous mutual information of the channel, i.e.,
$R_k=W\log_2(1+\frac{P_kH_k}{N_0W})$ (see (\ref{eq:rate})). As a
result, the transmitted packet could be decoded only when there is
exactly one user transmits. In order to allow for possibility of
capture, we set the data rate to be $\widetilde{R}_k=\beta
W\log_2(1+\frac{P_kH_k}{N_0W})$ in the simulation, where $\beta<1$.
As a result, we leave some margin in the transmit data rate so that
when there is collision, the transmit data rate may still be smaller
than the instantaneous mutual information
$C_k(\text{collision})=W\log_2(1+\frac{P_kH_k}{\sum_{i\neq
k}P_iH_i+N_0W})$ and packet detection is possible. The criteria to
determine the success of capture is based on comparing the
$\widetilde{R}_k$ and $C_k(\text{collision})$. If
$\widetilde{R}_k\leq C_k(\text{collision})$, then the packet from
the $k$-th user can be successfully decoded. Otherwise, it will be
corrupted.}. We set $\beta = 0.9$ to leave margin for the
possibility of capture in case of collision. It can be observed that
there is significant performance gain of the proposed scheme when
there is capture.

\section{Summary}\label{sec:summary}
We considered delay-sensitive transmit power and threshold control
design in S-ALOHA network. The users adaptively adjust their
transmission threshold and power, to achieve the minimal delay of
the network. The jointly optimal policy is revealed to be
computationally intractable and hence brute force solution is simply
infeasible. However, for a given threshold control policy, we
decompose the optimal power control policy into a reduced state MDP
for single user, in which the overall complexity is
$\mathcal{O}(NJ)$. Threshold control policy is proposed by
exploiting the special structure of the collision channel and the
common feedback to derive a low complexity solution, which is a one
dimensional optimization problem in symmetric and asymmetric
networks. The delay performance of the proposed design is
illustrated to have substantial gain relative to conventional random
access approaches in both networks.

\appendices
\section{Proof of Lemma~\ref{lem:state_tran}: Transition
Probability of Local System State}\label{app:tran} Note that the
transition event is from $\chi_{k,m}$ to
$\chi_{k,m+1}=\{Q_{k,m+1},H_{k,m},\gamma_{m},Z_{m},H_{k,m+1}\}$.
Specifically, the system threshold $\gamma_m$ is given by the
threshold control policy, i.e.,
$\gamma_m=\pi_{\gamma}(\gamma_{m-1},Z_{m-1})$ with certainty, and
$\Pr\{H_{k,m+1}=S_j|H_{k,m}=S_i\}=p_{i,j}$, independent of other
states. The transition probability of feedback and queue state is
given below.

\subsection{Feedback State Transition} \label{app:feedback}
From the position of user $k$, common feedback $Z_{m-1}$ and
$\{H_{k,m-1},\gamma_{m-1}\}$ could provide the information how many
other $(K-1)$ users have transmitted at the previous slot. It can be
ultilized to improve the prediction of their transmission behavior
at current slot. Moreover, whether user $k$ transmits at current
slot will influence the realization of $Z_m$ and hence, the feedback
transition is determined only by
$\{Z_{m-1},\{H_{k,i},\gamma_{i}\}_{i=m-1}^m\}$. Next we shall find
$\Pr\{Z_{m}|Z_{m-1},\{H_{k,i},\gamma_{i}\}_{i=m-1}^m\}$ (denote
$\Pr\{Z_{m}|Z_{m-1}\}$ for simplicity) given in
(\ref{eq:tran_local}).

In fact, the common feedback information could modify the stationary
probability of CSI states. For instance, $Z_{m-1}=0$ is equal to
$\bigcup_kH_{k,m}<\gamma_{k,m-1}$. Given $H_{k,m}<\gamma$, the
stationary probability $\Pr\{H_{k,m}=S_j\}$ should be modified as
$\widetilde{\pi}_{j}^k(\gamma)=\frac{\pi_j^k}{\sum_{S_i<\gamma}\pi_i^k}$.
Similarly, Given $H_{k,m}\geq\gamma$, the stationary probability
$\Pr\{H_{k,m}=S_j\}$ should be modified as
$\widehat{\pi}_{j}^k(\gamma)=\frac{\pi_j^k}{\sum_{S_i\geq\gamma}\pi_i^k}$.
Specifically, we introduce following definition for user $k$, where
$\gamma_{k,m}$ is the threshold for $k$-th user, utilized in section
\ref{sec:asymmetric}.

\begin{Def}[Transmission Event of the $k$-th User]
\label{def:transmission} Let $A_{k,m}$ denote the event that user
$k$ attempts to transmit at the $m$-th slot, i.e., $H_{k,m}\geq
\gamma_{k,m}$, while $\overline{A}_{k,m}$ denote the complimentary
event, i.e., $H_{k,m}<\gamma_{k,m}$. Furthermore, let
$B_{k,m}\in\{A_{k,m},\overline{A}_{k,m}\}$.
\end{Def}

As a result, the probability of the transmission event is given by:
\begin{equation}
\label{eq:tx_event}
\begin{array}{ll}\Pr\{A_{k,m}|\gamma_{k,m},\gamma_{k,m-1},B_{k,m-1}\}=\\
\left\{
\begin{array}{lll}
\upsilon_k=\sum\limits_{S_i<\gamma_{k,m-1}}\sum\limits_{S_j\geq\gamma_{k,m}}\widetilde{\pi}_{i}^k(\gamma_{k,m-1})p_{i,j}^k\text{if $B_{k,m-1}=\overline{A}_{k,m-1}$}\\
\zeta_k=\sum\limits_{S_i\geq\gamma_{k,m-1}}\sum\limits_{S_j\geq\gamma_{k,m}}\widehat{\pi}_{i}^k(\gamma_{k,m-1})p_{i,j}^k\text{if
$B_{k,m-1}=A_{k,m-1}$}
\end{array}\right.\end{array}
\end{equation}
For simplicity, let $\overline{\upsilon}_k=1-\upsilon_k$, and
$\overline{\zeta}_k=1-\zeta_k$. Note that, in symmetric network,
$\bigcup_{k}\upsilon_k=\upsilon$ and $\bigcup_{k}\zeta_k=\zeta$.
Therefore, we ignore the user index $k$ in the symmetric network.

\begin{itemize}
\item{\bf Feedback transits from $Z_{m-1}=0$ :}
All the other $(K-1)$ users did not transmit at the previous slot, and
transition probability is given by:
\begin{equation}
\begin{array}{ll}\Pr\{Z_m|Z_{m-1}=0\}=\\
\left\{ \begin{array}{ll}
\begin{array}{l}\overline{\upsilon}^{K-1}\mathbb{I}(\overline{A}_{k,m})\end{array}&\textrm{if $Z_{m}=0$}\\
\begin{array}{l}\overline{\upsilon}^{K-1}\mathbb{I}(A_{k,m})+\\
(K-1)\upsilon\overline{\upsilon}^{K-2}\mathbb{I}(\overline{A}_{k,m})\end{array}&\textrm{if $Z_{m}=1$}\\
\begin{array}{l}\left(1-\overline{\upsilon}^{K-1}\right)\mathbb{I}(A_{k,m})+\\
\left(1-\overline{\upsilon}^{K-1}-(K-1)\upsilon\overline{\upsilon}^{K-2}\right)\mathbb{I}(\overline{A}_{k,m})\end{array}
&\textrm{if $Z_{m}=e$}
\end{array} \right. \end{array}
\end{equation}

\item{\bf Feedback transits from $Z_{m-1}=1$:} Only one user's CSI exceeded $\gamma_{m-1}$ at
the previous slot, which could be divided into two cases.

If $H_{k,m-1}\geq\gamma_{k,m-1}$ ($A_{k,m-1}$ happens), all the
other users did not transmit at the previous slot. The CSI
information of other users are the same as $Z_{m-1}=0$ case, so the
transition probability is
$\Pr\{Z_{m}|Z_{m-1}=1,A_{k,m-1}\}=\Pr\{Z_{m}|Z_{m-1}=0\}$.

If $H_{k,m-1}<\gamma_{m-1}$ ($\overline{A}_{k,m-1}$ happens), only
one of other users transmitted at the previous slot. Then the
transition probability is given by:
\begin{equation}
\begin{array}{l}\Pr\{Z_m|Z_{m-1}=1,\overline{A}_{k,m-1}\}=\\
\left\{ \begin{array}{ll}
\begin{array}{l}\overline{\zeta}\overline{\upsilon}^{K-2}\mathbb{I}(\overline{A}_{k,m})\end{array}&\textrm{if $Z_{m}=0$}\\
\begin{array}{l}\overline{\zeta}\overline{\upsilon}^{K-2}\mathbb{I}(A_{k,m})+\\
\left(\zeta\overline{\upsilon}^{K-2}+\overline{\zeta}(K-2)\upsilon\overline{\upsilon}^{K-3}\right)\mathbb{I}(\overline{A}_{k,m})\end{array}&\textrm{if $Z_{m}=1$}\\
\left\{\left(1-\overline{\zeta}\overline{\upsilon}^{K-2}\right)\mathbb{I}(A_{k,m})+\right.&\textrm{if $Z_{m}=e$}\\
\begin{array}{l}\left(1-\overline{\zeta}\overline{\upsilon}^{K-2}-\zeta\overline{\upsilon}^{K-2}-\right.\\
\left.\left.\overline{\zeta}(K-2)\upsilon\overline{\upsilon}^{K-3}\right)\mathbb{I}(\overline{A}_{k,m})\right\}\end{array}
\end{array} \right.\end{array}
\end{equation}

\item{\bf Feedback transits from $Z_{m-1}=e$:}
At least two users transmitted at the previous slot, which should
also be divided into two cases. We first find the probability of
exact users involved in the transmission. Specifically, given
threshold $\gamma$, the probability that $k$ out $K$ users will
transmit is $\widetilde{p}_{\gamma}^{(K,k)}={K \choose
k}\left(\sum_{S_j\geq\gamma}\pi_{j}\right)^k\left(\sum_{S_j<\gamma}\pi_{j}\right)^{(K-k)}$.
Given additional information that at least $n$ users will transmit,
the probability is improved as
\begin{equation}
\label{eq:pr_e}
p_{\gamma,n}^{(K,k)}=\widetilde{p}_{\gamma}^{(K,k)}/\sum\nolimits_{i=0}^{n-1}\left(1-\widetilde{p}_{\gamma}^{(K,i)}\right),\forall
k\geq n
\end{equation}

If $H_{k,m-1}\geq\gamma_{m-1}$ ($A_{k,m-1}$ happens), at least one
of other users transmitted, i.e.,
\begin{equation}
\begin{array}{ll}\Pr\{Z_m|Z_{m-1}=e,A_{k,m-1}\}=\\
\left\{ \begin{array}{ll}
\sum\limits_{k=1}^{K-1}p_{\gamma,1}^{(K-1,k)}\left(\overline{\zeta}^{k}\overline{\upsilon}^{K-1-k}\right)\mathbb{I}(\overline{A}_{k,m})&\textrm{if $Z_{m}=0$}\\
\sum\limits_{k=1}^{K-1}p_{\gamma,1}^{(K-1,k)}\Bigl\{\left(\overline{\zeta}^{k}\overline{\upsilon}^{K-1-k}\right)\mathbb{I}(A_{k,m})+&\textrm{if $Z_{m}=1$}\\
\begin{array}{l}\bigl(k\zeta\overline{\zeta}^{k-1}\overline{\upsilon}^{K-1-k}\\
+\overline{\zeta}^{k}(K-1-k)\upsilon\overline{\upsilon}^{K-k-2}\bigr)\mathbb{I}(\overline{A}_{k,m})\Bigr\}\end{array}\\
\begin{array}{l}\sum\limits_{k=1}^{K-1}p_{\gamma,1}^{(K-1,k)}\Bigl\{\\
\left(1-\overline{\zeta}^{k}\overline{\upsilon}^{K-k-1}\right)\mathbb{I}(A_{k,m})\\
+\bigl(1-\overline{\zeta}^{k}\overline{\upsilon}^{K-k-1}-k\zeta\overline{\zeta}^{k-1}\overline{\upsilon}^{K-1-k}\\
-\overline{\zeta}^{k}(K-1-k)\upsilon\overline{\upsilon}^{K-k-2}\bigr)\mathbb{I}(\overline{A}_{k,m})\Bigr\}
\end{array}  &\textrm{if
$Z_{m}=e$} \end{array}\right. \end{array}
\end{equation}

If $H_{k,m-1}<\gamma_{m-1}$ ($\overline{A}_{k,m-1}$ happens), at
least two of other users transmitted, i.e.,
\begin{equation}
\begin{array}{ll}\Pr\{Z_m|Z_{m-1}=e,\overline{A}_{k,m-1}\}=\\
\left\{ \begin{array}{ll}
\sum\limits_{k=2}^{K-1}p_{\gamma,2}^{(K-1,k)}\left(\overline{\zeta}^{k}\overline{\upsilon}^{K-1-k}\right)\mathbb{I}(\overline{A}_{k,m})&\textrm{if $Z_{m}=0$}\\
\begin{array}{l}\sum\limits_{k=2}^{K-1}p_{\gamma,2}^{(K-1,k)}\Bigl\{\overline{\zeta}^{k}\overline{\upsilon}^{K-1-k}\mathbb{I}(A_{k,m})+\\
\bigl(k\zeta\overline{\zeta}^{k-1}\overline{\upsilon}^{K-1-k}+\\
\overline{\zeta}^{k}(K-1-k)\upsilon\overline{\upsilon}^{K-k-2}\bigr)\mathbb{I}(\overline{A}_{k,m})\Bigr\}\end{array}&\textrm{if $Z_{m}=1$}\\
\begin{array}{l}\sum\limits_{k=2}^{K-1}p_{\gamma,2}^{(K-1,k)}\Bigl\{\\
\left(1-\overline{\zeta}^{k}\overline{\upsilon}^{K-k-1}\right)\mathbb{I}(A_{k,m})\\
+\Bigl(1-\overline{\zeta}^{k}\overline{\upsilon}^{K-k-1}-k\zeta\overline{\zeta}^{k-1}\overline{\upsilon}^{K-1-k}\\
-\overline{\zeta}^{k}(K-1-k)\upsilon\overline{\upsilon}^{K-k-2}\Bigr)\mathbb{I}(\overline{A}_{k,m})\Bigr\}
\end{array} &\textrm{if $Z_{m}=e$} \end{array} \right. \end{array}
\end{equation}
\end{itemize}

\subsection{Queue State Transition} The queue state transition is correlated with the
feedback. For instance, if $Z_{m}\neq 1$, the probability of
decreased queue state should be zero, because of no successful data
receival. To obtain simple solution, we consider the case the same
as \cite{Vince:delay}, where the time slot duration $\tau$ is
substantially smaller than the average packet inter-arrival time and
average packet service time $\frac{1}{\mu}$ ($\tau \ll
\frac{1}{\lambda}$ and $\tau \ll \frac{1}{\mu}$), where $\mu$ is the
average packet service rate defined later.

\begin{itemize}
\item{\bf Packet arrival:} Since packet arrival follows Poisson
distribution with mean arrival rate $\lambda$, the transition
probability of the queue state related to packet arrival is given
by:
\begin{equation}
p_{q,q+1}=\Pr\{Q_{k,m+1}=q+1|Q_{k,m}=q\}=\lambda\tau
\end{equation}
\item{\bf Packet departure:}
The packet length follows exponential distribution with mean packet
size $\overline{N}_b$, so the packet service time also follows
exponential distribution. Conditioned on the state
$(\chi_{k,m},Z_m)$ and data rate given in (\ref{eq:rate}), the mean
packet service rate is:
\begin{eqnarray}
\label{eq:mu}
&&\mu(\chi_{k,m},Z_m,\pi_{P}(\chi_{k,m}))\\
&&=\frac{W}{\overline{N}_b}\log_{2}(1+\frac{P_{m}H_{k,m}}{N_0W})\mathbb{I}(Z_{m}=1)\nonumber
\end{eqnarray}
where $P_{k,m}=\pi_{P}(\chi_{k,m})$ is the power transmitted at
current slot determined by power control policy. Furthermore,
$Z_m\neq 1$ will lead to zero service rate. Another case leads to
zero service rate is $H_{k,m}<\gamma_{k,m}$, in which the power
control policy will set $P_{k,m}=0$. Hence, the probability for
packet departure is given by:
\begin{equation}
\begin{array}{l}
p_{q,q-1}=\\
\Pr\{Q_{k,m+1}=(q-1)^+|Q_{k,m}=q,\chi_{k,m},Z_{m},\pi_{P}(\chi_{k,m})\}\\
=\mu(Q_{k,m}=q,\chi_{k,m},Z_{m},\pi_{P}(\chi_{k,m}))\tau
\end{array}
\end{equation}
\item{\bf No change in the $k$-th user:} The transition probability
corresponding to no change in queue state is given by:
\begin{equation}
\begin{array}{l}
p_{q,q}=\\
\Pr\{Q_{k,m+1}=q|Q_{k,m}=q,\chi_{k,m},Z_{m},\pi_{P}(\chi_{k,m})\}\\
=(1
-p_{q,q-1}- p_{q,q+1})
\end{array}
\end{equation}
\end{itemize}

Since $\lambda\tau\ll 1$ and $\mu\tau\ll 1$, the probability of
multiple packet arrivals or packet departures is negligible and
hence $p_{q,p}=0$ for $|p-q|>1$. Thus the transition probability of
queue state is given by
$\Pr\{Q_{k,m+1}|\chi_{k,m},Z_{m},\pi_{P}(\chi_{k,m})\}$, which
completes the proof.

\section{Proof of Lemma \ref{lem:unichain}: Decidability of the Unichain of Reduced
State} \label{app:unichain} Denote the state (excluding $Q$) in
$\hat{\chi}$ as $\Phi=(H,\gamma,Z)$, whose transition probability
has been given in lemma \ref{lem:state_tran}, independent of $Q$ and
power control policy. Specifically,
$\Pr\{\Phi_{m+1}|\Phi_{m}\}=\Pr\{H_m|H_{m-1}\}\Pr\{Z_m|\Phi_{m},H_{m},\gamma_{m}\}$,
where $\gamma_m$ is determined from the given threshold control
policy. Then, the recurrent classes of $\Phi$ could be found.
Furthermore, the queue state evolves as a birth-death process under
every power control policy, forming an unichain itself. As a result,
the unichain of the reduced state $\hat{\chi}=(Q,\Phi)$ is
decidable.

\section{Proof of Theorem~\ref{thm:complexity}: Complexity of the Reduced State
MDP}\label{app:complexity} $\{\zeta_k,\upsilon_k\}$ in
$(\ref{eq:tx_event})$ are functions of
$\{\gamma_{m-1},\gamma_{m}\}$. Specifically, we assume
$\zeta_k\geq\upsilon_k$ for the same $\{\gamma_{m-1},\gamma_{m}\}$.
This is a practical assumption for fading channels, because the CSI
states will not change fast\cite{FSMC:Channel}. Then, we have
following lemma about the threshold control policy $\gamma_m^*$ in
(\ref{eq:gamma_symm}).
\begin{Lem}[Monotonic Increasing Function of $\gamma_m^*$ w.r.t $K$]
\label{lem:Non-decreasing} Given $\{\gamma_{m-1},Z_{m-1}\}$, if
$K_2\geq K_1$,
$\gamma_m^*(\gamma_{m-1},Z_{m-1},K_2)\geq\gamma_m^*(\gamma_{m-1},Z_{m-1},K_1)$.
Specifically, if
$\gamma_m^*(\gamma_{m-1},Z_{m-1},K_1)<\gamma_{m-1}$, then for a
sufficiently large $K_2$,
$\gamma_m^*(\gamma_{m-1},Z_{m-1},K_2)>\gamma_m^*(\gamma_{m-1},Z_{m-1},K_1)$.
\end{Lem}
\begin{proof}
Given $\{\gamma_{m-1},Z_{m-1}\}$, $\gamma_m$ just influence the
$\{\zeta,\upsilon,\overline{\zeta},\overline{\upsilon}\}$ parameter
in (\ref{eq:gamma_symm}), and from (\ref{eq:tx_event}),
$\{\zeta,\upsilon\}$ are monotonic decreasing
($\{\overline{\zeta},\overline{\upsilon}\}$ are monotonic
increasing) functions of $\gamma_m$. As a result, lemma
\ref{lem:Non-decreasing} is obvious when $Z_{m-1}\neq e$. If
$Z_{m-1}=e$, when $K_1$ is increased to $K_2$, by comparing each
term of the same $k$ case and in additional $k=(K_1+1)\cdots K_2$
case, using the assumption of $\zeta\geq\upsilon$ for the same
$\{\gamma_{m-1},\gamma_{m}\}$, the monotonic increasing
characteristic is also obvious.
\end{proof}

As the reduced state is $\hat{\chi}=\{Q,H,\gamma,Z\}$, the worst
case complexity is corresponding to the total number of states of
$\hat{\chi}$, i.e., $O(NJ^2)$. On the other hand, since the QSI and
CSI states are recurrent, the least number of states in a recurrent
class is $\mathcal{O}(NJ)$. Next we will show that the number of
states of the system threshold $\gamma$ decreases as $K$ increases,
and $\gamma=S_J$ regardless of the feedback when $K$ is large
enough, which completes the proof.

Given $K_1$, let $\gamma_{\min}(K_1)$ be the minimal threshold in a
recurrent reduced state class. Specifically,
$\gamma_{\min}(K_1)=\gamma_m^*(\gamma_{K_1},Z_{K_1},K_1)$, where
$\gamma_{K_1}>\gamma_{\min}(K_1)$. By lemma
\ref{lem:Non-decreasing}, for a sufficiently large $K_2>K_1$,
$\gamma_{\min}(K_2)>\gamma_{\min}(K_1)$ and hence, the minimal
threshold in the recurrent class is increased. Following the
argument, the minimal threshold will increase to the largest CSI
state $S_J$, when $K$ is increased to a large number $K_0$.

\bibliographystyle{IEEEtran}
\bibliography{IEEEabrv,Huang}

\begin{IEEEbiography}[{\includegraphics[width=1in,height=1.25in,clip,keepaspectratio]{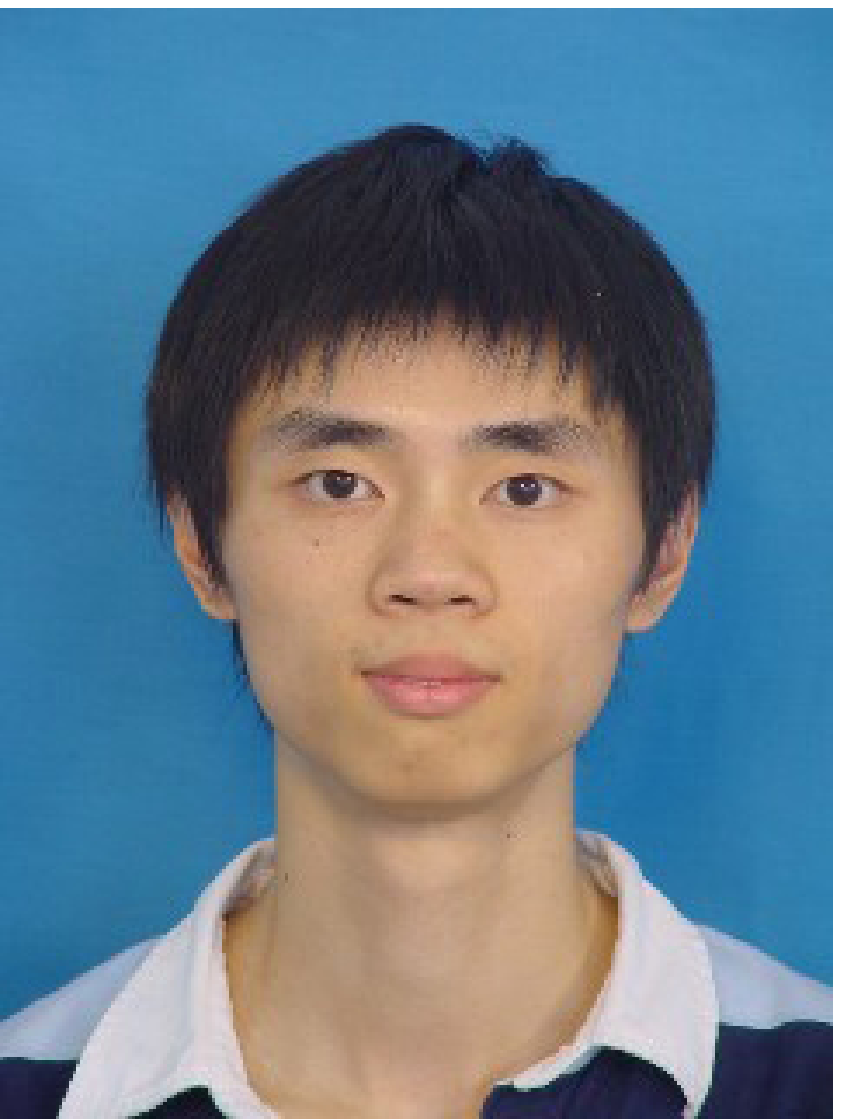}}]{Huang Huang}
received the B.Eng. and M.Eng. (Gold medal) from the Harbin
Institute of Technology(HIT) in 2005 and 2007 respectively, all in
Electrical Engineering. He is currently a PhD student at the
Department of Electrical and Computer Engineering, The Hong Kong
University of Science and Technology. His recent research interests
include cross layer design and performance analysis via game theory
in random access network, and embedded system design.
\end{IEEEbiography}
\begin{IEEEbiography}[{\includegraphics[width=1in,height=1.25in,clip,keepaspectratio]{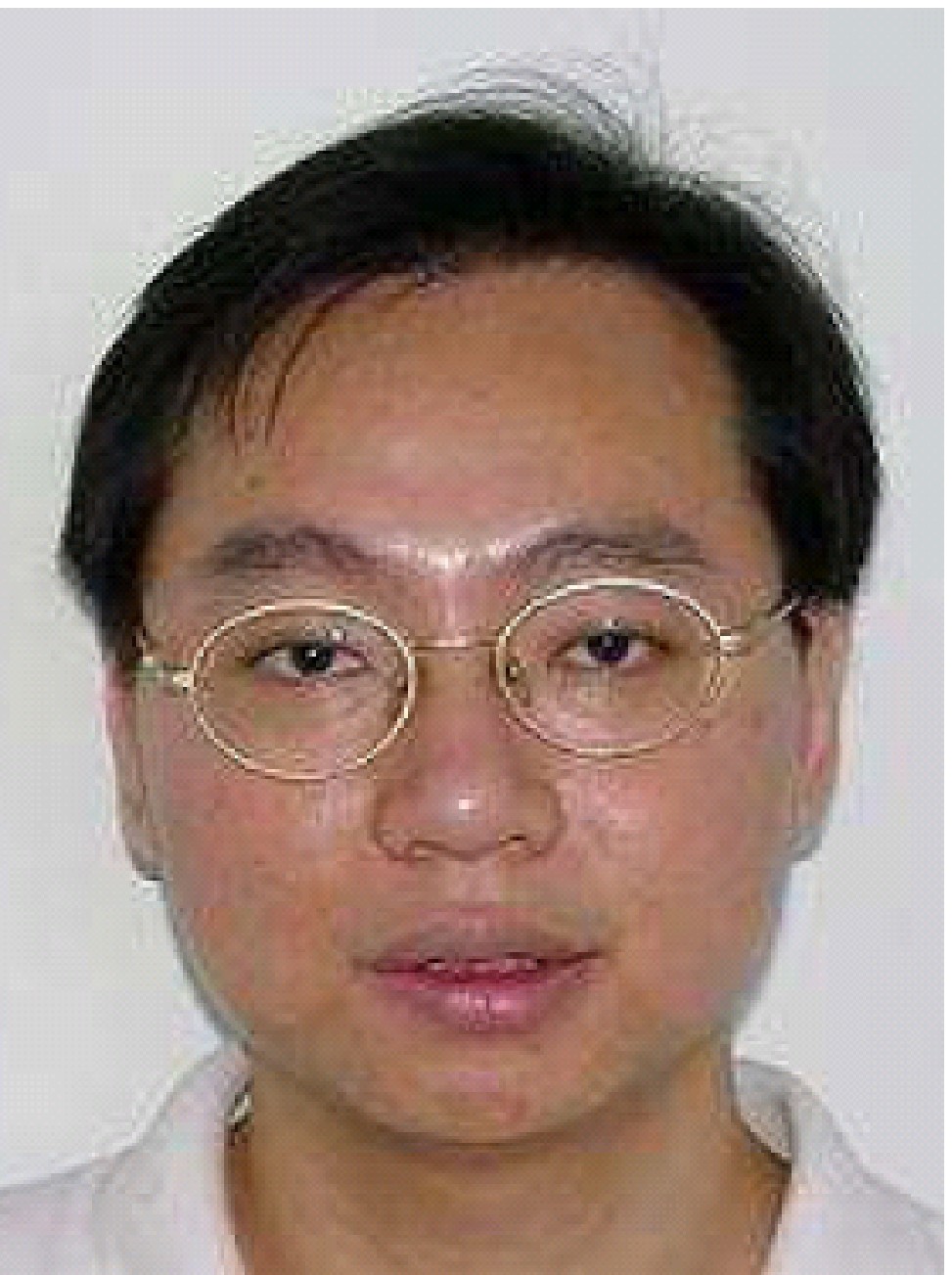}}]{Vincent Lau}
obtained B.Eng (Distinction 1st Hons) from the University of Hong
Kong (1989-1992) and Ph.D. from Cambridge University (1995-1997). He
was with HK Telecom (PCCW) as system engineer from 1992-1995 and
Bell Labs - Lucent Technologies as member of technical staff from
1997-2003. He then joined the Department of ECE, Hong Kong
University of Science and Technology (HKUST) as Associate Professor.
His current research interests include the robust and
delay-sensitive cross-layer scheduling of MIMO/OFDM wireless systems
with imperfect channel state information, cooperative and cognitive
communications, dynamic spectrum access as well as stochastic
approximation and Markov Decision Process.
\end{IEEEbiography}

\end{document}